\documentclass[12pt,reqno,a4paper]{amsart}
\usepackage{amssymb,latexsym}
\usepackage{amsthm, amsmath, amsfonts}
\usepackage{hyperref}
\usepackage{empheq}
\usepackage{graphicx}  
\usepackage{dcolumn}   
\usepackage{bm}        
\usepackage{empheq}
\usepackage{enumerate}
\usepackage{slashed}
\usepackage{simplewick}
\usepackage{comment}
\usepackage{color}
\usepackage{soul}
\usepackage[english]{babel}
\usepackage{appendix}
\usepackage[utf8]{inputenc}

\usepackage{marginnote}
\usepackage[normalem]{ulem}

\def\blue{\color{blue}}

\def\be{\begin{equation}}
\def\ee{\end{equation}}
\def\bea{\begin{eqnarray}}
\def\eea{\end{eqnarray}}
\def\bz{\bar z}

\def\bs{\bar s}
\def\p{\partial}
\def\bp{\bar\partial}

\def\bs{\bar s}
\def \CO{\mathcal O}

\def\bcal_k{\mathcal B_k}
\newcommand{\szego}{Szeg\H{o}\ }

\newcommand{\kahler}{K\"ahler }
\newcommand{\om}{\omega}
\newcommand{\ov}{\overline}
\newcommand{\cC}{\mathcal{C}}
\def\tr{{\rm Tr\,}}

\DeclareMathOperator{\ric}{Ric}
\DeclareMathOperator{\End}{End}
\newcommand{\boldsym}[1]{\boldsymbol{#1}}
\newcommand\bb{\boldsym{b}}

\newcommand{\field}[1]{\mathbb{#1}}
\newcommand{\Z}{\field{Z}}
\newcommand{\R}{\field{R}}
\newcommand{\C}{\field{C}}

\newtheorem{thm}{{\sc Theorem}}
\newtheorem{maindefn}{{\sc Definition}}

\newtheorem{mainlem}{{\sc Lemma}}

\newenvironment{example}
{\medskip\noindent{\it Example:\/} }{\medskip}
\newenvironment{rem}
{\medskip\noindent{\it Remarks:\/} }{\medskip}

\newtheorem*{mycor}{Corollary}

\newtheoremstyle{dotless}{}{}{\itshape}{}{\bfseries}{}{ }{}
\theoremstyle{dotless}
\newtheorem*{mythm}{Theorem}

\begin{document}

\title[Quantum Hall effect and Quillen metric]
{Quantum Hall effect and Quillen metric}
\author[Semyon Klevtsov, Xiaonan Ma,
George Marinescu and Paul Wiegmann]
{Semyon Klevtsov$^1$, Xiaonan Ma$^2$,
George Marinescu$^{3}$ and Paul Wiegmann$^4$}

\maketitle

\begin{center}
{\small
\address{\it$^1$Universit{\"a}t zu K{\"o}ln,
Mathematisches Institut, Weyertal 86-90, 50931 K{\"o}ln, Germany}

\vspace{.15cm}
\address{\it$^2$Universit\'e Paris Diderot - Paris 7,
UFR de Math\'ematiques, Case 7012,\\
75205 Paris  Cedex 13, France}

\vspace{.15cm}
\address{\it$^3$Universit{\"a}t zu K{\"o}ln,
Mathematisches Institut, Weyertal 86-90, 50931 K{\"o}ln, Germany\\
\&Institute of Mathematics `Simion Stoilow', Romanian Academy,
Bucharest, Romania}

\vspace{.15cm}
\address{\it$^4$Department of Physics, University of Chicago,
929 57th St, Chicago, IL 60637, USA}
}

\vspace{.3cm}
\email{\tt\footnotesize sam.klevtsov@gmail.com,
xiaonan.ma@imj-prg.fr,\\ gmarines@math.uni-koeln.de,
wiegmann@uchicago.edu}
\end{center}

\vspace{.5cm}
\begin{abstract} We study the generating functional,
the adiabatic curvature and the adiabatic phase for the integer
quantum Hall effect (QHE) on a compact Riemann surface.
For the generating functional we derive its asymptotic expansion
for the large flux of the magnetic field, i.e., for the large degree $k$
of the positive Hermitian line bundle $L^k$. The expansion consists
of the anomalous and exact terms. The anomalous terms are the
leading terms of the expansion. This part  is responsible for the
quantization of the adiabatic transport coefficients in QHE.
We then identify the non-local (anomalous) part of the expansion 
with the Quillen metric on the determinant line bundle, and the
subleading exact part with the asymptotics of the regularized
spectral determinant of the Laplacian for the line bundle $L^k$,
at large $k$. Finally, we show how the generating functional of
the integer QHE is related to the gauge and  gravitational (2+1)d
Chern-Simons functionals. We observe the relation between
the Bismut-Gillet-Soul\'e curvature formula for the Quillen metric
and the adiabatic curvature for the  electromagnetic and geometric
adiabatic transport of the integer Quantum Hall state.
Then we relate the adiabatic phase in QHE to the eta invariant and show that the geometric part of the adiabatic phase is given by the Chern-Simons functional.
\end{abstract}

\maketitle

\thispagestyle{empty}
\tableofcontents

\section{Introduction}


Quantum Hall effect (QHE) arises as a result of a collective motion
of electronic liquid confined to the plane in a quantizing magnetic
field.  Over the last thirty years, several mathematical models of
QHE have been developed. In early important papers
\cite{ASZ,ASZ1,L1,L2} a fundamental  relation between QHE
and the Quillen metric \cite{Q}  has been used to illuminate 
quantization of transport coefficients.
Since then this relation has not been sufficiently  explored. 
This paper intends to fill this gap.
We further develop the relation between Quillen metric and QHE,
and demonstrate that the anomaly formula for Quillen metric
(sometimes the name ``Quillen anomaly'' is used for brevity)
is the powerful approach to physics
and mathematics of QHE and is a particularly useful approach to
capture the universal
transport coefficients. The relation between the Quillen anomaly
and QHE also allows us to obtain the asymptotic expansion
of the spectral determinant of  the Laplacian in large magnetic field.

The adiabatic transport coefficients are perhaps
the most important observables of QHE.
They are precisely quantized on QH plateaus. These coefficients are
encoded by the adiabatic curvature 2-form, see e.g.,  \cite{AS}.
We will show that the adiabatic curvature 2-form  and through it
the quantized transport coefficients,
naturally follow from the formula \cite{BGS} (see also \cite{B,BF})
for the curvature of the
determinant line bundle.  Also we will show how the curvature
formula can be used in order
to obtain the geometric part of the adiabatic phase of the QHE in 
the form
of the gauge and gravitational Chern-Simons functional.
The geometric part of the adiabatic phase has been recently
computed in papers  \cite{AG1,AG2,AG4}, extending
the earlier approaches of \cite{FS} and \cite{WZ}. Our method in
this paper is based on generating functional
(see, e.g., \cite{ASZ,L1,L2,K,FK,CLW,CLW1,KW} for an incomplete
list of references). One of the outcomes of this article is
a clarification of the relation between the generating functional
QH-state approach, the Quillen anomaly and the Chern-Simons
functional. In this paper we limit our studies to the integer QHE.

\subsection{Background on QHE}\label{PIntr}
Quantum Hall effect occurs in certain two-dimensional electron 
systems subjected to a high magnetic field. If the current $J$ 
is forced through the sample in $x$-direction, the Hall voltage 
$V_H$ is observed in the perpendicular direction $y$ on the plane. 
The Hall conductance $\sigma_H=J/V_H$, measured as 
a function of magnetic field at a fixed chemical potential,
follows a series of transitions between the plateaus,
where it takes constant rational values $\sigma_H=\nu$ in units of
$e^2/h$ with a very high precision. 
The case where $\nu\in \mathbb Z$ is referred as the integer 
QHE and $\nu=p/q$ is referred to as the fractional QHE. 
The former case can be explained by a system of free 
(non-interacting) electrons and the latter is strongly interacting 
system (via Coulomb forces). While the mechanism behind the QHE
is completely different in the two cases, the quantization of 
the transport coefficients 
is explained by invoking their
relations to  certain Chern numbers. 
This paradigm is traced back to 
Ref. \cite{T1}  and was elaborated in more detail in
Refs.\ \cite{ASZ,ASZ1} and \cite{L1,L2}. Let us briefly recall 
the main points of their construction. 

The idea is to consider the QHE on a compact Riemann surface 
$\Sigma$, where the collective wave function $\Psi(z_1,...,z_N)$ 
describing the electrons naturally acquires dependence on 
the parameters associated with the geometry of $\Sigma$. 
To give an example, we need to recall that in quantum mechanics 
of particle in a magnetic field, the wave function is a section 
of an appropriate line bundle $\rm L$ on $\Sigma$, 
see e.g.\ \cite{DK} for a review. The flat component (called $L_\varphi$ in what follows) of the line bundle 
$\rm L$ is parameterized by coordinates 
$(\varphi_1^a,\varphi_2^a)\in T^{2\rm g}$ in the Jacobian torus
(see Eq.\ \eqref{flatcon} for more detail), which physically 
correspond to Aharonov-Bohm (AB) solenoid fluxes, 
piercing through the holes of the surface (for genus $\rm g>0$). 
The wave function $\Psi(z_1,...,z_N|\varphi^a)$ now depends on 
$\varphi^a$, which can be varied adiabatically with time $t$ 
and we can define adiabatic (Berry) connection 
\begin{equation}\label{aconn}
\mathcal A=i\langle \Psi, d_t\Psi\rangle_{L^2}
\end{equation} 
and curvature $\Omega=\langle d_t\Psi, d_t\Psi\rangle_{L^2}$, 
see Ref.\ \cite{Simon}. The current $J_a$ around the $a$th flux 
is equal $J_a=i\omega_{ab}\dot\varphi^b$, where $\omega_{ab}$
are coefficients of the conductance two-form $\omega$ on
$T^{2\rm g}$, computed via Quillen anomaly arguments in 
Ref. \cite{ASZ1},
\begin{equation}\label{sigmaH}
\omega=\Omega+dd^c\log\det \Delta_{\rm L},\quad
\Omega= \sigma_H\sum_{a=1}^{\rm g}
d\varphi_1^a\wedge d\varphi_2^{a},
\end{equation}
where $\Delta_{\rm L}$ is the Laplacian for the line bundle $\rm L$, defined 
in \eqref{DeltaL}. Thus the Hall conductance $\sigma_H$, 
which is the coefficient in the conductance  2-form $\Omega$
in \eqref{sigmaH} is quantized since $[\Omega]\in2\pi\mathbb Z$, 
and $dd^c\log\det \Delta_{\rm L}$ can be interpreted 
as the mesoscopic fluctuations around the quantized value of 
$\sigma_H$. Mesoscopic fluctuations vanish exponentially 
with the size of the system and its average over  the ensemble of 
systems with respect to remaining parameters identically vanishes.

In this paper  we focus on the adiabatic geometric transport on 
the moduli space of complex 
structures on $\Sigma$, which gives rise to novel quantized 
transport coefficients \cite{ASZ1,L1,L2,KW}. 
In this paper we explore this program further, 
along various lines. Before stating the results, we introduce 
the main objects of interest.

\subsection{Mathematical background}\label{M}

Consider a positive Hermitian holomorphic line bundle $L$ with
a Hermitian metric $h$,
and its $k$th tensor power $(L^k,h^k)$, on a compact (connected)
Riemann surface
$\Sigma$ of genus $\rm g$ and the space of global holomorphic
sections $H^0(\Sigma, L^k)$.
Physically, the basis of the vector space $H^0(\Sigma, L^k)$
corresponds to the states on the lowest Landau level (LLL).
We then  tensor the line bundle $(L^k,h^k)$ with the
$s$-th tensor power $K^s$ of the
canonical line bundle $K$ on $\Sigma$,
where $s\in\frac12\mathbb Z$ is called the spin.
Note that the canonical bundle $K$ on a Riemann surface
admits a square root, namely, there exists a
holomorphic line bundle $\mathcal{K}$ such that $\mathcal{K}^2=K$,
so that $K^s:=\mathcal{K}^{2s}$ makes sense for
$s\in\frac12\mathbb Z$.

For the line bundle $L^k$ the positive integer $k$ has a meaning
of total flux of the magnetic field  $\frac1{2\pi}\int_\Sigma F=k$,
where $-iF=F_{z\bz}dz\wedge d\bz$,
$F_{z\bz}=-\p_z\p_{\bz}\bigl(\log h^k \bigr)$,
is the curvature $(1,1)$-form of $L^k$
and we assume that $F_{z\bz}>0$ everywhere on $\Sigma$.
Writing the
Riemannian metric as
$ds^2=2g_{z\bz}|dz|^2$ in local complex coordinates $z,\bz$,
the Hermitian metric on $K^s$ is  $(g_{z\bz})^{-s}$.
Then $B=g^{z\bz}F_{z\bz}$ 
is the magnetic field.
The curvature form on $K^s$ is given by $si\,{\rm Ric}(g)$,
where the Ricci form is
${\rm Ric}(g)=-\p_z\p_{\bz}\bigl(\log g_{z\bz}\bigr)idz\wedge
d\bz$.
Thus $\deg K^s=-s\chi(\Sigma)$ and according to the Riemann-Roch
formula \eqref{RR1}
the dimension of the vector space $H^0(\Sigma,L^k\otimes K^s)$
of holomorphic sections (LLL states) equals
\begin{equation}\nonumber
N_k=\dim H^0(\Sigma,L^k\otimes K^s)=k+(1-{\rm g})(1-2s),
\end{equation}
assuming $\deg L^k\otimes K^s>\deg K$. Since the physics of
QHE involves
a large number of particles, we will always assume that the flux
of the magnetic field is large $k\gg1$,
while $s$ is kept fixed, so that the dimension $N_k$ is also large.

Consider now a basis $s_j(z),\,j=1,\ldots N_k$, of the space of
holomorphic sections
$H^0(\Sigma,L^k\otimes K^s)$.
The positive measure on configurations of $N_k$ electrons derived
from QH-state is defined as follows: 
\begin{maindefn}\label{L}
Consider the following finite measure on a configuration of $N_k$
points
\begin{equation}\label{meas}
|\Psi(z_1,\ldots,z_{N_k})|^2\prod_{j=1}^{N_k} \sqrt gd^2z_j
:= \frac1{N_k!}|\det s_i(z_j)|^2 \prod_{j=1}^{N_k} h^k(z_j,\bz_j)
g_{z\bz}^{-s}(z_j,\bz_j)\sqrt gd^2z_j\,.
\end{equation}
The partition function $Z_k$ is defined as
\begin{equation}\label{Z}
Z_k=\frac1{(2\pi)^{N_k}}\int_{\Sigma^{N_k}}
|\Psi(z_1,\ldots,z_{N_k})|^2\prod_{j=1}^{N_k} \sqrt gd^2z_j,
\end{equation}
{\blue } and $\log Z_k$ is called the generating functional.
\end{maindefn}
By construction $Z_k$ depends on  the choice of the
basis in $H^0(\Sigma,L^k\otimes K^s)$, which we specify in 
a moment, on the  Riemannian
metric $g$ on $\Sigma$,
the Hermitian metric $h^k$ on $L^k$, and is invariant
if we transform the defining data by a diffeomorphism. 
Also, $Z_k$ implicitly depends on the choice
of a complex structure $J$
on $\Sigma$ and a point in Jacobian variety. 
The latter is because once we have chosen the line bundle $L^k$, 
all other choices are related to by tensoring 
$L^k\otimes L_\varphi$ with a flat line bundle $L_\varphi$ 
(which is trivial as a smooth line bundle, but non-trivial as 
holomorphic line bundle).
Following standard notations we denote by $\mathcal M_{\rm g}$ of complex structures on the surface $\Sigma$ of genus $\rm g$ and by $\mathit{Jac}(\Sigma)$
the moduli space of flat line bundles on $\Sigma$. 
Then $Z_k$ varies over the parameter space (complex manifold),
\begin{equation}\label{YY}
Y=\mathcal{M}_{\rm g}\times \mathit{Jac}(\Sigma).
\end{equation}
Note that the complex structure on the fiber
$\mathit{Jac}(\Sigma_{b})$ depends on $b\in \mathcal{M}_{\rm g}$.
For each $y\in Y$ we denote by $\Sigma_y$ the Riemann surface
corresponding to
$y$, and we have a corresponding partition function $Z_k(y)$.
As we have already mentioned in \S \ref{PIntr}, physically 
the Jacobian variety $\mathit{Jac}(\Sigma)$
corresponds to the space of Aharonov-Bohm solenoid fluxes.

The dependence of generating functional on $y\in Y$ encodes
the adiabatic transport coefficients related to electromotive
($\mathit{Jac}(\Sigma)$) and geometric ($\mathcal M_{\rm g}$)
transport, as we now explain.
Consider the holomorphic line bundle
$\mathcal L=\det H^0(\Sigma, L^k\otimes K^s)$ over $Y$, called
determinant line bundle, with fiber
$\mathcal L_y=\det H^0(\Sigma_y, L_y^k\otimes K_y^s)$
for $y\in Y$. For each point $y_0\in Y$ there exist a neighborhood
$U\subset Y$ of $y_0$
and a basis $s_j=s_{j}(\cdot,y)$, $j=1,\ldots,N_k$, of
$H^0(\Sigma_y, L_y^k\otimes K_y^s)$
holomorphically varying with $y\in U$.
We obtain thus a local holomorphic frame
\begin{equation}\label{Phi1}
\mathcal S:U\to \mathcal L,\quad y\mapsto 
\mathcal S(y):=s_{1}(\cdot,y)\wedge\cdots
\wedge s_{N_{k}}(\cdot,y)\in\mathcal{L}_y\,.
\end{equation}
This is the holomorphic component of the state $\Psi$
in Eq.\  \eqref{meas}, which also varies holomorphically on
the parameter space $Y$. We define now $Z_{k}(y)$ Eq.\ \eqref{Z}
for $y\in U$ as the square of the $L^2$-norm 
of $\mathcal{S}(y)$ 
with help of $s_{j}(\cdot,y)$, by taking the pointwise norm
of $s_{1}(\cdot,y)\wedge\cdots\wedge s_{N_{k}}(\cdot,y)$
as a section of 
$(L_y^k\otimes K_y^s)\boxtimes\ldots
\boxtimes(L_y^k\otimes K_y^s)\to\Sigma^{N_k}_y$
and integrating over $\Sigma^{N_k}_y$. 
In Lemma 1 (Section \ref{sec5}) 
we show that the adiabatic
curvature $(1,1)$-form $\Omega$ on $Y$ is given by the curvature 
of the Chern connection on $\mathcal{L}$ endowed with 
the $L^2$-metric,
\begin{equation}\label{omega1}
\Omega=-\partial_Y \overline\partial_{Y} \log Z_{k},
\end{equation}
where $\overline\partial_{Y}$ is the Cauchy-Riemann operator
on $Y$ and $\partial_Y$ is defined by the decomposition
$d_Y=\partial_Y+\overline\partial_{Y}$ of the exterior derivative
$d_Y$ on $Y$.
In local complex coordinates $(y_j,\overline{y}_j)$ on $Y$,
$\partial_Y=\sum_jdy_j\wedge\partial_{y_j}$,
$\overline\partial_{Y}=\sum_jd\overline{y}_j
\wedge\partial_{\overline{y}_j}$,
where
$\partial_{y_j}=\partial/\partial y_j$,
$\partial_{\overline{y}_j}=\partial/\partial\overline{y}_j$\,.
Therefore the generating functional $\log Z_k$ plays the role of
a \kahler potential on $Y$. Note that there is a freedom 
in the definition of $Z_k$, since we can always multiply 
$Z_k\to |f(y)^2|Z_k$, where $f(y)$ is 
a non-vanishing holomorphic function on $Y$. In particular, 
$Z_k$ is not necessary modular invariant on 
$\mathcal M_{\rm g}$. However, this ambiguity will be localized at 
the boundary of $Y$ and will not contribute to the adiabatic 
curvature Eq.\ \eqref{omega1} in the bulk, where the results
of this paper apply. Additional arguments, similar to that of 
Ref. \cite{BK1}, 
can be invoked to fix this ambiguity, which we do not consider here.
The integrals of $\Omega$ over
smooth closed two-cycles in $Y$ define quantized adiabatic
transport coefficients. Here we observe the following relation
\begin{equation}\label{omega2}
\Omega=\Omega^{\mathcal L}-\partial_Y \overline\partial_{Y}
{\det}'\Delta_{\rm L}\,,
\end{equation}
between the adiabatic curvature $\Omega$, the regularized
spectral determinant of the (Kodaira) Laplacian
$\Delta_{\rm L}={\bp_{\rm L}}^*\bp_{\rm L}^{\phantom{a}}$
for the line bundle ${\rm L}=L^k\otimes K^s$ (the asterisk in
$\bp_{\rm L}^*$ denotes the adjoint operator) and the curvature
of the Quillen metric $\Omega^{\mathcal L}$, to be defined next.
Since $\partial_Y \overline\partial_{Y} \log\det'\Delta_{\rm L}$
is an exact form on $Y$, \eqref{omega2} shows that
$\Omega$ and $\Omega^{\mathcal L}$
belong to the same de Rham cohomology class.

Following Quillen \cite{Q} we define a smooth Hermitian metric
 on $\mathcal L$, called Quillen metric,
such that the norm squared of the section
$\mathcal S$ is given by Eq.\  \eqref{detsec},
\begin{equation}\label{Smetr}
\|\mathcal S\|^2=\frac{\det \langle s_j,s_i\rangle_{L^2}}
{\det'\Delta_{\rm L}},
\end{equation}
where $L^2$ metric is given by 
Eqns.\ \eqref{hernorm},\eqref{inner},
since from the determinant formula \eqref{detformula},
it follows that
$\det \langle s_j,s_i\rangle_{L^2}=Z_k$.
By general theory, the Chern curvature of the Quillen metric
is given on $U$ by
\begin{equation*}
\Omega^{\mathcal L}=-\partial_Y \overline\partial_{Y}
\log\|\mathcal S\|^2\,.
\end{equation*}
The curvature of
the Quillen metric
is then the following $(1,1)$-form on $Y$,
\begin{equation}\label{Qano}
\Omega^{\mathcal L}=-\partial_Y \overline\partial_{Y}
\log \frac{Z_k}{\det'\Delta_{\rm L}}\,,
\end{equation}
and the relation \eqref{omega2} follows immediately.

Quillen \cite{Q} observed that the metric in \eqref{Smetr} is
a smooth metric on the determinant line bundle $\mathcal L\to Y$,
where $Y$ is the space of holomorphic structures of
a complex vector bundle on a fixed Riemann surface
and computed precisely the curvature of its Chern connection.
Adapting Quillen's result yields the formula for the curvature
$\Omega^{\mathcal L}$ in the case of $Y=\mathit{Jac}(\Sigma)$.
If $Y = \mathcal M_{\rm g}$ is the moduli space 
for Riemann surfaces, $\Omega^{\mathcal L}$ was
considered by Belavin-Knizhnik \cite{BK1},
and for a smooth family of Dirac operators by
Bismut-Freed \cite{BF} and (for higher-dimensional \kahler manifolds)
by Bismut-Gillet-Soul\'e \cite{BGS}, see also \cite{AMV}
for the account in physics literature.
Note that in general the cohomology class of
$\Omega^{\mathcal L}$ is non-trivial, as a consequence of
Atiyah-Singer family index theorem \cite{AtS}.

The adiabatic phase, or holonomy 
$\exp(- {\int_{\mathcal C}\mathcal A^{\mathcal L}})$ is the phase 
factor acquired by the wave 
function under the transport along a smooth closed contour
$\mathcal C$ in $Y$.   The adiabatic phase consist of 
two  distinct parts: the topological part and the geometric part. 
The  topological part, arises if the contour $\mathcal C$ is 
non-contractible, such as e.g., the Dehn twists on the torus, 
and is independent of the smooth variations of the contour.  
Such contour encloses a boundary point in $Y$ where 
the holomorphic function $f(y)$ discussed above is singular.
This part is due to a flat connection on $\mathcal L\to Y$, 
which has a non-zero holonomy.  
The  geometric part of the adiabatic phase arises even if 
the contour is contractible and depends on the shape of 
the contour. The sum of two phases  equal full adiabatic phase, and if the contour is contractible, the geometric phase is the total phase.

Next, following \eqref{omega2}, we will consider 
the connection 1-form, associated to the curvature of 
the Quillen metric
 $\Omega^{\mathcal L}=d_Y\mathcal A^{\mathcal L}$.
Then $\mathcal A^{\mathcal L}-\mathcal A
 =\frac12(\overline\partial_Y-\partial_Y)
 \log\det'\Delta_{\mathrm{L}}$
 is a globally defined differential form on $Y$, related to the
 Bismut-Cheeger eta form \cite[Definition\,4.33]{BC},
 \cite[(3.166)]{BF}, \cite[Theorem\,2.20]{BGSII}.
The holonomy $\exp( -{\int_{\mathcal C}\mathcal A^{\mathcal L}})$
 on a closed contour $\mathcal C$ is related to the
 $\eta$-invariant \cite{APS,W}, \cite[Theorem\,3.16]{BF},
 \cite[\S6]{BB}.
 In Th. \ref{prop2} we will compute the part of the 
geometric adiabatic phase 
$\int_{\mathcal C}\mathcal A^{\mathcal L}$, associated to 
the Quillen metric.

One of the central goals of this article is to determine the large $k$
asymptotics of the generating functional $\log Z_k$,
to determine its dependence on the geometric parameters and
relate it to the asymptotics of the Quillen metric and to compute
the transport coefficients and the geometric adiabatic phases.
As a consequence of Eq.\ \eqref{Smetr}, we also obtain the
asymptotics of the regularized spectral determinant and compare
it to the previous results of Bismut-Vasserot \cite{BV}.

\subsection{Main results}

The first result concerns the large $k$ asymptotic expansion of
the generating functional $\log Z_k$. We  show that
the asymptotic expansion has schematically the form
$\log Z_k=\log  Z_H+\mathcal F$, where $\log  Z_H$
is the non-local (also referred to as ``anomalous'') 
part of the expansion, and
the ``exact part'' $\mathcal F$
contains local terms. The anomalous part is completely defined
by the Quillen metric \eqref{Smetr}.

Consider the metrics $g_0,g$ and $h_0^k$, $h^k$, related as
$g=g_0+\p_z\p_{\bz}\phi$ and
$h^k=h_0^ke^{-k\psi}$, for scalar functions
$\phi,\psi\in \mathcal C^{\infty}(\Sigma)$
cf.\ \eqref{kclass}, \eqref{kclass1}.
We study the large $k$ asymptotics of 
$\log \frac{Z_k}{Z_{k0}}$, where $Z_{k0}$
is defined as in Eq.\ \eqref{Z} using the metrics $g_0$ and
$h_0^k$.
\begin{thm}[\S\ref{sec31}]\label{prop1}
At large $k$ the following asymptotic expansion holds 
\begin{equation}\nonumber
\log \frac{Z_k}{Z_{k0}}=\log  \frac{Z_H}{Z_{H0}}
+\mathcal F-\mathcal F_0,
\end{equation}
where the anomalous part $\log Z_H$ is given by
\begin{align}\label{freeintro}
\log  \frac{Z_H}{Z_{H0}}=&\frac 2{\pi}\int_\Sigma\left( (A_z
+\frac{1-2s}2\omega_z)(A_{\bz}+\frac{1-2s}2\omega_{\bz})
-\frac1{12} \omega_z\omega_{\bz}\right) d^2z\\ \nonumber
&- (A\to A_0,\,\omega\to \omega_0),
\end{align}
(see also Eq.\ \eqref{ano} for another form of this equation),
where the notation $(A\to A_0,\,\omega\to \omega_0)$ just means that we replace
$A$ by $A_0$ and $\omega$ by $\omega_0$
in the previous expression,
and the exact part $\mathcal F$ admits a large $k$ asymptotic
expansion, with the first three terms given by
\begin{align}
\label{freeintro1}
\mathcal F=-\frac1{2\pi}\int_\Sigma\left[\frac12B\log \frac B{2\pi}
+\frac{2-3s}{12}R\log\frac{B}{2\pi}+\frac1{24}(\log B)
\Delta_g(\log B)\right]
\sqrt gd^2z+\mathcal O(1/k),
\end{align}
where $\Delta_g=2g^{z\bz}\p_z\p_{\bz}$ is the scalar Laplacian.
\end{thm}
Moreover,
\begin{thm}[\S\ref{sec4}]\label{prop11}
The anomalous part of the generating functional corresponds to
the Quillen metric \eqref{Smetr},
\begin{equation}\nonumber
\log Z_H=\log\|\mathcal S\|^2,
\end{equation}
and the exact part
corresponds to the spectral determinant
\begin{displaymath}
\log\frac{{\det}'\Delta_{\rm L}}{{\det}'\Delta_{{\rm L}0}}
=\mathcal{F}-\mathcal{F}_0\,.
\end{displaymath}
\end{thm}
Let us first explain the content of Eqs.\ \eqref{freeintro},
\eqref{freeintro1}.
Here
$A_z,\,A_{\bz}$ are the components of the connection 1-form
for the magnetic field strength $F=dA$,
and $\omega_z,\,\omega_{\bz}$ are the components of the
spin connection 1-form
${\rm Ric}(g)=d\omega$, see \eqref{gauge2}, \eqref{gauge1}
for details. Also, $A_0$ and $\omega_0$ refer to the same objects 
defined with respect to the background metrics $g_0$ 
and $h_0^k$. The  objects are weighted as follows:  
$B$ (and $A_z,A_{\bz}$) are of order $k$,
meaning $B/k$ is a smooth positive function on $\Sigma$, 
independent of $k$,
while the objects that depend on the metric, such as $\Delta_g$,
$R$ and $\omega_z,\omega_{\bz}$ are considered to be of order
$1$. Hence the remainder ${1}/{k}$ terms in \eqref{freeintro1} may
contain terms of the type  $1/B$.

The proof of Theorem \ref{prop1} is based on the known
asymptotic expansion of the Bergman kernel.
The variational formula \eqref{freeen} connects the large
$k$ asymptotic expansion
of $\log Z_k$ with the large $k$ asymptotic expansion of
the Bergman kernel for positive Hermitian holomorphic line bundle
$L^k$, established in Refs.\ \cite{Z,C}. The more general version of
the Bergman kernel expansion pertinent to the present paper
was given in Ref.\ \cite{MM,MM1}.
We also refer to \cite{MM} for a comprehensive study of several
analytic and geometric aspects of Bergman kernel.
Suppose $\{s_j'\}$ is a basis of
$H^0(\Sigma,L^k\otimes K^s)$, orthonormal with respect to
the inner product
Eq.\ \eqref{inner}. Then the Bergman kernel on the diagonal
$B_k(z,\bz)$ is defined by the sum
\begin{align}\nonumber
B_k(z,\bz)=&\sum_{j=1}^{N_k}|s_j'(z)|^2h^kg_{z\bz}^{-s}
\label{B5},
\end{align}
The first three terms of
its large $k$ asymptotic expansion of the Bergman kernel on the
diagonal are\begin{align}
B_k(z)=B+\frac{1-2s}4R+\frac14\Delta_g\log B+\mathcal O(1/k)\,.
\end{align}
In  Eq.\ \eqref{expansion} we list more terms in this expansion, relevant for
the result of Theorem\ \ref{prop1}, and in the Appendix B we review the Bergman kernel expansion.

Integrating the variational formula \eqref{freeen}, we find that
the generating functional can be separated in two parts.
The anomalous part \eqref{freeintro} is a non-local functional of
the magnetic field density $B$ and scalar curvature $R$.
This part is given in Eq. \eqref{freeintro}.
The exact part $\mathcal F$ of the asymptotic expansion consists
of local functionals of $B$ and $R$. Using Theorem \ref{prop11}
we interpret $\mathcal F$ as
$\log\det'\Delta_{\rm L}$ and  explain that all
corrections of order $1/k$ to $\mathcal F$ are in fact given by
local functionals.

Then, in Theorem \ref{prop11} we identify the anomalous part of
the expansion $\log Z_H$
with the Quillen metric \eqref{Smetr}. We use the heat kernel
formula for the regularized
determinant $\det'\Delta_{\rm L}$ of the Laplacian for
the line bundle ${\rm L}=L^k\otimes K^s$,
in order to derive the anomaly formula for
$\log\|\mathcal S\|^2$ under the variations of \kahler
and magnetic potentials. Let us point out that this is consistent
with the general anomaly formula
for $\log\|\mathcal S\|^2$ of Ref.\ \cite[Theorem\,1.23]{BGS}.
As a consequence, we note that the exact
part $\mathcal F$ \eqref{freeintro1} of the expansion can be
identified with asymptotics of
$\log \det'\Delta_{\rm L}$ for high powers $k$ of the positive
line bundle.
The leading term in this expansion was derived by
Bismut-Vasserot \cite{BV},
see also \cite[\S 5.5.5]{MM}. Thus Eq.\ \eqref{freeintro1}
provides the leading corrections to their result.

Theorem \ref{prop1} generalizes the large $k$ asymptotic
expansion of the generating functional derived
in Ref.\ \cite{K}, where the case of constant
magnetic field $B=k$ and $s=0$ was considered. For analogous
results in fractional QHE, derived by various physics methods
see Refs.\ \cite{ZW,CLW,CLW1,FK}.
In particular the exact terms in Eq.\ \eqref{freeintro1}
are in agreement with \cite[Eq.\ (130)]{CLW1}.


In the last section we derive the formula for the geometric part
of the adiabatic phase, associated with the curvature form 
$\Omega^{\mathcal L}$. The starting point is the 
Bismut-Gillet-Soul\'e curvature
 Eq.\ \eqref{GRR}\footnote{In the QHE context this formula was
 invoked in Ref.\ \cite{TP06}.} (\cite[Theorems\,1.9, 1.27]{BGS},
see also \cite{BF, BJ}) 
\begin{equation}\label{BFformula}
\Omega^{\mathcal L}=-2\pi i\int_{M|Y}
\bigl[{\rm Ch}(E){\rm Td}(TM|Y)\bigr]_{(4)}.
\end{equation}
Here $\sigma:M\to Y$ is a fibration over $Y$ (universal curve), with
fiber diffeomorphic to $\Sigma$,
and $\Sigma_y=\sigma^{-1}(y)$ for $y\in Y$.
Also,
\begin{equation} \label{Ek}
    E = \tilde{L}^k \otimes \tilde{K}^s \to M\,,
    \end{equation}
where the line bundles $\tilde{L}^k, \tilde{K}^s$ have been
extended from the individual
Riemann surfaces $\Sigma_y$ to the universal curve,
see \S \ref{sec5} for details.
Moreover, ${\rm Ch}(E)$ is the Chern form and ${\rm Td}(TM|Y)$ is
the Todd form and the
integrand in Eq.\ \eqref{BFformula} is a 4-form on the manifold $M$.
The integration
goes over the fiber $\Sigma_y$ over
a fixed point $y\in Y$, so that the result is a 2-form on $Y$.
We introduce the curvature
$\mathrm{F}^{E}=\mathrm{F}-s\mathrm{R}_{TM|Y}$ where
$\mathrm{F}$  refers to the part of the curvature 2-form
corresponding to the line bundle $\tilde L^k\to M$.
Locally on $M$, ${\rm F} = d A$, with a connection 1-form $A$,
and $\mathrm{R}_{TM|Y}=d\omega$ with
a spin-connection 1-form $\omega$.

Then using the curvature
formula \eqref{BFformula} we show that the geometric part
(induced by $\Omega^{\mathcal L}$) of the adiabatic phase
is given by the Chern-Simons action.
\begin{thm}[\S\ref{sec52}]\label{prop2}
The adiabatic phase arising from the Quillen metric over a closed contour
$\mathcal C\subset Y$ takes the form 
    \begin{align}\label{eq:gCS5}
\exp\Big(- {\int_{\mathcal C}\mathcal A^{\mathcal L}}\Big)
    = \exp\Big(2i\pi\int_{\sigma^{-1}(G)}
 \bigl[{\rm Ch}(E){\rm Td}(TM|Y)\bigr]_{(4)}\Big),
\end{align} 
if $\mathcal{C}$ is the boundary of a domain $G\subset Y$.

For a general closed contour $\mathcal C\subset Y$,  if the metric $g$ on $T\Sigma$
is the restriction of a K\"{a}hler metric 
on a neighborhood  of $\sigma^{-1}(\mathcal C)$, the adiabatic phase is given by
\begin{align}\label{eq:gCS6}
\exp\Big(- {\int_{\mathcal C}\mathcal A^{\mathcal L}}\Big)
    = \exp\big( 2i\pi \overline{\eta}\big),
\end{align}
where $\overline{\eta}$ is the adiabatic limit of the
reduced eta invariant for the fibration 
$\sigma: \sigma^{-1}(\mathcal C) \to \mathcal C$.
If $\overline{\eta}_1$ is the adiabatic limit
of the reduced eta invariant  associated to metrics
$g_1$ and $h^L_1$ on $T\Sigma$ and $L$, 
the variation of the adiabatic limit is given by 
the (2+1)d Chern-Simons functional
\begin{multline}\label{curcm6.201}
2\pi i(\overline{\eta}-\overline{\eta}_{1})
=\frac{i}{4\pi}\int_{\sigma^{-1}(\mathcal{C})} A\wedge dA
+\frac{1-2s}2(A\wedge d\omega+dA\wedge\omega)+
\left(\frac{(1-2s)^2}4-\frac1{12}\right)\omega\wedge d\omega\,\\
- (A\to A_{1}, \omega\to \omega_{1}).
\end{multline}
\end{thm}
For a general contour $\mathcal{C}$ and general metric $g$
we can obtain a similar formula to 
\eqref{eq:gCS6} by using the anomaly formula \cite[Theorem 1.23]{BGS}
(cf.\ Theorem \ref{thm2}, \eqref{thm21}) as in the proof of \cite[Theorem 1.27]{BGS}.

For $s=0$ this formula is equivalent to the Chern-Simons action
of QHE derived in Ref.\ \cite{AG1}.
Let us also point out that results of this paper admit generalization
to the Laughlin states in fractional QHE with the filling fraction
$\nu=1/\beta$, where $\beta$ is an integer,
as was reported in Ref.\ \cite{KW}. These ``$\beta$-deformations''
of the Quillen anomaly and Bismut et.\ al.\ curvature formulas will
be explored further elsewhere.

\subsection{Further context}

It was understood early on that the exact quantization of the
transport coefficients
in QHE essentially has a geometric origin
\cite{T1,T2,AS,ASZ,ASZ1}. In particular, the quantization of
the Hall conductance in the integer QHE case was attributed to
the fact, that the QH-wave functions are sections of a line bundle
over the space of Aharonov-Bohm fluxes
(Jacobian variety $\mathit{Jac}(\Sigma)$) on a higher-genus
Riemann surface,
and the corresponding adiabatic curvature belongs to
an integer Chern class \cite{ASZ}. In the seminal papers
\cite{ASZ1} and \cite{L1} the existence of another quantized
parameter in the integer QHE -- so called odd
viscosity (terms anomalous viscosity or Hall viscosity are also
used) -- was deduced when considering the adiabatic curvature
of integer QH-states for the adiabatic transport on the moduli space
of complex structures of a torus. For the moduli spaces of complex
structures of higher-genus Riemann surfaces the adiabatic curvature
in the integer QHE was first derived in the pioneering paper
\cite{L2}. These results were also generalized
to the fractional Quantum Hall states, for the quantization of
the conductance \cite{AS,TW} and for the anomalous viscosity
on the torus \cite{TV1,TV2,R,RR}. Existence of yet another,
third quantized coefficient in both, integer and the fractional
QHE associated with complex structure moduli spaces of
higher-genus Riemann surfaces was proposed in \cite{KW},
see \cite{BR1} for related results.

As we have emphasised, many aspects of the geometry of Quantum Hall
states are encoded in the generating functional 
$\log Z_k$ \eqref{Z}.
The geometric definition of $Z_k$ in QHE context in terms of 
holomorphic sections of
the line bundle and \kahler metrics on the Riemann surface was
given in \cite{K}. There exist several approaches to calculate
the generating functional. In \cite{ZW} and \cite{CLW,CLW1}
the  Ward identity method was applied in case when $\Sigma$
is the sphere. In \cite{K}, as well as in
the present work, the large $k$ expansion is obtained from
asymptotic expansion of the Bergman kernel \cite{Z,C,MM} 
on the surface of any genus.
Interestingly, objects closely related to the QHE generating
functional are useful in  \kahler geometry \cite{Don2,B1,B2}.
The field theory approach to calculate the large $k$ asymptotics
of the generating functional was developed in \cite{FK}, and the
collective field theory approach was proposed in \cite{LCW}.

Alternatively, in the case of the integer QHE the transport
coefficients can be also directly computed through the linear
response theory as it has been done in \cite{AG1}.
The linear response theory can be also cast in the
form of the so-called effective action approach.
Although conceptually different,  the effective action of
the LLL and the geometric adiabatic phase discussed in this paper
are identical in the case of the QHE. Both are given by the
Chern-Simons functional. For other relevant approaches to
the effective action which also lead to
the Chern-Simons functional see \cite{AG4,Son,BR,BR1}.

The generating functional is formally equivalent to thermodynamic
potential of 2d Coulomb plasma of equal charges in a neutralized
background, the so-called plasma analogy.
From this point of view the normalization factor is
a partition function of the 2d Coulomb gas
(see, e.g., \cite{For} for the comprehensive account).
This relation gives allows us to establish the link between
the thermodynamics of 2d Coulomb plasma, Quillen metric and
the spectral determinant of the Laplacian.

Finally, let us remark that the most notable occurrence of
the Quillen metric in physics
goes back to the work of Belavin-Knizhnik \cite{BK1}
on the holomorphic anomaly of the
Polyakov string. Cancellation of the
holomorphic anomaly yields the result that the
measure of the critical string is a modulus square of the
holomorphic function on the moduli space. Interestingly,
in contrast to the critical string, in QHE the
curvature of Quillen metric does not vanish and furthermore 
encodes
important physical information (quantized adiabatic transport).
\\

\noindent
{\bf Acknowledgments.}
We would like to thank J.-M.~Bismut, V.~Pestun and S.\ Zelditch for useful discussions and 
comments on the manuscript and the anonymous referees for
useful comments and suggestions. 
We thank Zhenghan Wang for pointing us Ref.\ \cite{Kirby89}.
SK is partially supported by 
the Max Delbr\"uck prize for junior researchers at 
the University of Cologne, the Humboldt postdoctoral
fellowship and the grants NSh-1500.2014.2,  RFBR 15-01-04217
and  NSF DMR-1206648.
XM is partially supported by ANR-14-CE25-0012-01
and funded through the Institutional Strategy of
the University of Cologne within the German Excellence Initiative.
GM was partially supported by the DFG project SFB TR12.
The research of PW was carried out at the RAS Institute
for Information Transmission Problems under support by a grant
from the Russian Foundation for Sciences (project 14-50-00150).
SK, XM and GM acknowledge the hospitality of the Simons Center
for Geometry and Physics
where the research leading to this paper has been started.

\section{Generating functional and determinant line bundle}
\label{freefer}
\subsection{Wave functions on the lowest Landau level}
\label{sec21}

We begin by defining the wave functions on the lowest Landau level,
on a compact (connected)
Riemann surface $\Sigma$. Let $(L,h)$ be a positive Hermitian
holomorphic line bundle
and let $(L^k,h^k)$ be its $k$th tensor power, with
$k$ a positive integer.
In local complex coordinates $z,\bz$ the metric on $\Sigma$
is diagonal
$ds^2=2g_{z\bz}|dz|^2$ and the area of the surface will be fixed
$\int_\Sigma\sqrt g d^2z=2\pi$. The curvature $(1,1)$-form of
the Hermitian metric is $-iF$, where
\begin{align}\nonumber
F= (\p_zA_{\bz}-\p_{\bz}A_z)dz\wedge d\bz
=-\p_z\p_{\bz}(\log h^k) idz\wedge d\bz
=:F_{z\bz}idz\wedge d\bz. 
\end{align}
Here  $A=A_zdz+A_{\bz}d\bz$ is the local gauge connection
1-form for $F$ and
\begin{align}
\label{gauge2}
&A_z=\frac12i\p_z\log h^k,\quad A_{\bz}
=-\frac12 i\p_{\bz}\log h^k.
\end{align}
The positivity of $(L,h)$ means that
$F$ is a positive (1,1)-form on $\Sigma$,
or equivalently, the function $F_{z\bz}=-\p_z\p_{\bz}\log(h^k)$
is positive.
Instead of the $(1,1)$-form $F$ we will mainly use the magnetic field
$B=g^{z\bz}F_{z\bz}$. It is a globally defined scalar function,
also everywhere  positive on $\Sigma$.
Then the integer tensor power $k$ has the meaning of the total flux
of the magnetic field through the surface
\begin{equation}\label{e:flux}
\frac1{2\pi}\int_\Sigma F
=\frac1{2\pi}\int_\Sigma B\sqrt gd^2z=k.
\end{equation}
Let us also introduce the spin-connection 1-form $\omega$,
related to the Ricci (1,1)-form ${\rm Ric}(g)$
as follows 
\begin{align}\nonumber
{\rm Ric}(g)=R_{z\bz}\, idz\wedge d\bz
=(\p_z\omega_{\bz}-\p_{\bz}\omega_z) dz\wedge d\bz,
\end{align}
and
\begin{align}\label{gauge1}
&\omega_z=\frac12i\p_z\log g_{z\bz},\quad\omega_{\bz}
=-\frac12 i\p_{\bz}\log g_{z\bz}.
\end{align}
The scalar curvature $R$ is given by
\begin{equation}\label{def1}
R=2g^{z\bz}R_{z\bz}=-2g^{z\bz}\p_z\p_{\bz}\log\sqrt g
=-\Delta_g\log\sqrt g,
\end{equation}
where $\sqrt g=2g_{z\bz}$ and the scalar Laplacian is
$\Delta_g=2g^{z\bz}\p_z\p_{\bz}$.
We use the standard normalization for the scalar curvature,
so that the Euler characteristic of $\Sigma$ is given by
$\chi(\Sigma)=\frac1{4\pi}\int_\Sigma R\sqrt gd^2z$.

The LLL wave functions are in one-to-one correspondence with
the holomorphic sections of the
line bundle. In our case the line bundle will be the tensor product
${\rm L}=L^k\otimes K^s$,
where $K$ is canonical line bundle. The Hermitian metric on
the sections of
$K^s$ is $g_{z\bz}^{-s}$ and its curvature form is $is{\rm Ric}(g)$.
The holomorphic sections are solutions of the $\bp$-equation
\begin{equation}
\label{dbareq}
\bp_{\rm L} s(z)=0,
\end{equation}
where
the $\bp$-operator acts as follows
\begin{equation}\label{acts}
\bp_{\rm L}: \mathcal C^\infty(\Sigma,L^k\otimes K^s)
\to\Omega^{0,1}(\Sigma,L^k\otimes K^s)
\end{equation}
from $\mathcal C^\infty$ sections of $L^k\otimes K^s$ to
$(0,1)$ forms with coefficients in $L^k\otimes K^s$. 
For background on the $\bp$-equation
and holomorphic sections we refer to Ref.\ \cite{MM}.
Solutions of Eq.\ \eqref{dbareq} form a vector space
$H^0(\Sigma,L^k\otimes K^s)$.
The holomorphic Euler characteristic of $L^k\otimes K^s$
is calculated by the Riemann-Roch theorem
\cite[p.\,245-6]{GH:78}, \cite[Theorem 1.4.6]{MM},
\[
\chi(L^k\otimes K^s):=\dim H^0(\Sigma,L^k\otimes K^s)-
\dim H^1(\Sigma,L^k\otimes K^s)
=\deg(L^k\otimes K^s)+1-\mathrm{g}\,.
\]
Recall that $\deg K =2(\mathrm{g}-1)$
and by our assumption \eqref{e:flux},
$\deg L=\int_{\Sigma}c_1(L)=1$.
Note that a holomorphic line bundle $L$ is positive on $\Sigma$
if and only if $\deg L>0$, cf.\ \cite[p.\,214]{GH:78}.
On the other hand, by the Kodaira vanishing theorem,
$H^1(\Sigma,L^k\otimes K^s)=0$ if
$\deg L^k\otimes K^s>\deg K$,
see \cite[p.\,215]{GH:78}, equivalently,
if $k+ 2 (\mathrm{g}-1)(s-1) >0$.
Hence for $s\in\frac12\mathbb Z$ and
$k>2 (1-\mathrm{g})(s-1)$ we have
\begin{equation}\label{RR1}
N_k=\dim H^0(\Sigma,L^k\otimes K^s)=k+(1-{\rm g})(1-2s)\,.
\end{equation}
In this paper we assume that the magnetic field flux $k\gg1$,
so that the number of states is always positive and large.

In the symmetric gauge \eqref{gauge2}, \eqref{gauge1},
the LLL wave functions solve a closely related equation
\begin{equation}\label{dbar}
\bar D\psi_j=0,
\end{equation}
where the locally-defined operator $\bar D$ is nothing 
but the global operator $\bp_{\rm L}$ in \eqref{dbareq}, 
written in the chosen gauge. Namely
\begin{equation}\label{barD}
\bar D=  h^{\frac k2} g_{z\bz}^{-\frac{s+1}2}
i\bp_L\bigl( h^{-\frac k2} g_{z\bz}^{\frac s2}\; \cdot\bigr)
= g_{z\bz}^{-\frac12} (i\p_{\bz}-s\omega_{\bz}+A_{\bz}),
\end{equation}
The half-integer $s$ is called the spin. Consequently,
Eq.\ \eqref{dbar} can be formally solved by twisting the basis
in $H^0(\Sigma,L^k\otimes K^s)$ as follows
\begin{equation}\label{basiswf}
\mathcal\psi_j(z,\bz)= s_j(z) h^{\frac k2}g_{z\bz}^{-\frac s2},
\end{equation}
where $s_j(z)$ is a holomorphic section of $L^k\otimes K^s$,
so the number of solutions of \eqref{barD} is the same as
in Eq.\ \eqref{RR1}.

Norm squared of a section at a point $z\in\Sigma$ is defined
with the help of the Hermitian metric, and is given by
\begin{equation}\label{hernorm}
\|s(z)\|^2=|s(z)|^2h^k(z,\bz)g_{z\bz}^{-s}(z,\bz),
\end{equation}
Now $\|s(z)\|^2$ is a scalar function on $\Sigma$, which can be
integrated. More generally, the $L^2$ inner product of sections
reads
\begin{equation}\label{inner}
\langle s_1,s_2\rangle=\frac1{2\pi}\int_\Sigma \bs_1s_2\,
h^kg_{z\bz}^{-s}\sqrt gd^2z.
\end{equation}
This immediately translates into the standard quantum-mechanical
inner product for the  wave functions \eqref{basiswf},
\begin{equation}\nonumber
\langle\psi_1|\psi_2\rangle
=\frac1{2\pi}\int_\Sigma \psi_1^*\psi_2\sqrt gd^2z.
\end{equation}
Now we are ready to define the generating functional in the
integer Quantum Hall.

\subsection{Generating functional}
\label{sec22}

Consider a non-normalized basis $\{s_j\}$ of
$H^0(\Sigma,L^k\otimes K^s)$ (later it will be important that 
the sections $s_j$ depend on $Y$ holomorphically) and
the corresponding basis of wave functions \eqref{basiswf}.
Then there exists a section of the determinant line bundle
$(L^k\otimes K^s)_{z_1}\wedge\cdots\wedge
(L^k\otimes K^s)_{z_{N_k}}$, 
associated to the basis $s_j$.
(More precisely, the line bundle is
$\det (\oplus_j \pi^*_j (L^k\otimes K^s))$ on $\Sigma^{N_k}$,
where $\pi_j$ is the projection on the $j$-th factor of
$\Sigma^{N_k}$.)
This section is given by the completely antisymmetric combination
$s_1\wedge\ldots\wedge s_{N_k}$ and can be formally written as
a Slater determinant
\begin{align}
\label{slater}
\frac1{N_k!}\det \bigl[s_j(z_l)\bigr]_{j,l=1}^{N_k}.
\end{align}
This is a collective wave function of free electrons on LLL,
which we refer to as the integer QH-state. Physically it means
that the lowest Landau level is completely filled, i.e., we have
$N_k$ particles (fermions) at the positions $z_1,\ldots, z_{N_k}$.
The Hermitian norm of the section \eqref{slater} is induced by
the point-wise norm \eqref{hernorm} for each $s_j$ and reads
\begin{equation}\nonumber
|\Psi(z_1,\ldots,z_{N_k})|^2:=\frac1{N_k!}|\det s_j(z_l)|^2
\prod_{j=1}^{N_k}\,h^k(z_j,\bz_j)g_{z\bz}^{-s}(z_j,\bz_j).
\end{equation}
This is a scalar function on $\Sigma^{N_k}$, which defines
a measure on configuration of $N_k$ points \eqref{meas}.
The partition function \eqref{Z} is
\begin{align}\label{genfunc}
Z_k=\frac1{(2\pi)^{N_k}N_k!}\int_{\Sigma^{N_k}}|\det s_j(z_l)|^2
\prod_{j=1}^{N_k} h^k(z_j,\bz_j) g_{z\bz}^{-s}(z_j,\bz_j)
\sqrt gd^2z_j,
\end{align}
and $\log Z_k$ is called the generating functional.

Another useful representation of the partition function is in terms
of the determinant of the Gram matrix of inner products
\eqref{inner} of sections
\begin{equation}\label{detformula}
Z_k=\det\frac1{2\pi} \int_\Sigma \bs_j(\bz)s_l(z) 
h^kg_{z\bz}^{-s}\sqrt gd^2z=\det \langle s_j,s_l\rangle.
\end{equation}
This representation is called the determinantal formula. From this
representation it  follows that the integral \eqref{genfunc}
converges, since the sections are $L^2$-integrable \eqref{inner}.

By construction $Z_k$ depends on the choice of the basis in
$H^0(\Sigma,L^k\otimes K^s)$. However, this dependence
is not hard to take care of. Note that under the linear change
of the basis of sections $s_i\to s'_i=A_{ij}s_j$ the generating
functional transforms as
\begin{align}\label{detaa}\nonumber
&s_i\to s'_i=A_{ij}s_j,\quad A\in GL(N_k,\mathbb C)\\
&Z_k\to\det (A^*A)\cdot Z_k.
\end{align}
Since  $A_{ij}$ is a numerical matrix, the derivatives of $\log Z_k$
are $A$-independent, hence they are independent of the choice of
the basis of sections.
Consider now some reference metrics $h_0^k$ on $L^k$ and
$g_{0}$ on $\Sigma$, connected with $h^k$ and $g$ as
\begin{align}\label{kclass}
&h^k(z,\bz)=h_0^k(z,\bz)e^{-k\psi(z,\bz)}\\
&g_{z\bz}=g_{0z\bz}+\p_z\p_{\bz}\,\phi(z,\bz),\label{kclass1}
\end{align}
where the $\mathcal C^{\infty}$ scalar functions $\psi(z,\bz)$
and $\phi(z,\bz)$ are known as  the \kahler potential and the
magnetic potential, respectively. In particular, Eq.\ \eqref{kclass}
translates into
\begin{equation}\label{def2}
F_{z\bz}=F_{0z\bz}+k\p_z\p_{\bz}\psi,
\end{equation}
Since the metric and the magnetic field are positive,
$g_{z\bz}>0,\,F>0$, the functions $\psi(z,\bz)$ and
$\phi(z,\bz)$ satisfy
$\p_z\p_{\bz}\phi(z,\bz)>-g_{0z\bz}$ and $\p_z\p_{\bz}\psi>
-F_{0z\bz}/k$.

Consider now the generating functional $\log Z_{k0}$, defined
using the reference metrics. It follows from the transformation rule,
Eq.\ \eqref{detaa}, that the difference
\begin{equation}\label{Z1}
\log Z_k-\log Z_{k0}
\end{equation}
is independent of the choice of the basis in
$H^0(\Sigma,L^k\otimes K^s)$.

Finally, let us explain why $\log Z_k$ is called the generating
functional.
Writing $h^k=h_0^k e^{-k\psi}$,  it is straightforward to see that
the variational derivatives of \eqref{Z} with respect to $\psi$
\begin{equation}\label{dens}
\frac{\delta}{\delta\psi(w_1,\bar w_1)}\cdots
\frac{\delta}{\delta\psi(w_m,\bar w_m)}\log Z_k
=(-k)^m\big\langle\rho(w_1,\bar w_1)\ldots\rho(w_m,\bar w_m)
\big\rangle_c
\end{equation}
give the connected multi-point correlation functions of electronic
density operator
$\rho(z,\bz)=\sum_{j=1}^{N_k} \delta(z,z_j)$.
These density correlation functions were considered
in \cite{ZW,CLW,CLW1}.

\subsection{Determinant line bundle and Quillen metric}
\label{Y}

In \cite{Q} Quillen defined the metric on determinant line bundle,
using regularized determinants of the Cauchy-Riemann operators
on Riemann surfaces.
It was used for the computation of holomorphic anomalies in
string theory amplitudes by Belavin-Knizhnik \cite{BK1}.
More generally, the curvature formulas for
 were derived for determinant lines of
Dirac operators by Bismut-Freed \cite[Theorem\,1.21]{BF}
and for \kahler manifolds
by Bismut-Gillet-Soul\'e \cite[Theorem\,1.9]{BGS}.
For the relevant literature see also
the following physics \cite{AMV,VV,DP} and
mathematical papers \cite{BJ,ZT}.
In the context of the integer quantum Hall effect on
Riemann surfaces the Quillen metric was first discussed in
Refs.\ \cite{ASZ,L2}, see also \cite{KW} for the fractional case.

In Appendix A we give a general overview of the Quillen metric and 
analytic torsion, and here we consider the Riemann surface case,
which is of main relevance to this paper.

Consider again the $\bp_L$ operator \eqref{dbareq}, which acts
from the vector space of sections to the vector space of
$(0,1)$-forms \eqref{acts}. The adjoint operator
$\bp_L^*$ relative to the inner
product \eqref{inner} acts in the opposite direction.
Hence the (Kodaira) Laplacian can be defined as
\begin{equation}\label{DeltaL}
\Delta_{\rm L}=\bp_L^*\bp_L^{\phantom{a}}:\;
\mathcal C^{\infty}(\Sigma,L^k\otimes K^s)
\to\mathcal C^{\infty}(\Sigma,L^k\otimes K^s).
\end{equation}
It will be convenient to work with
its Hermitian-equivalent. Namely, we consider the operator $\bar D$
as defined in \eqref{barD}. The adjoint operator $\bar D^*$ reads
\begin{align}\nonumber
\bar D^*=g_{z\bz}^{-\frac12}(i\p_z-(s-1)\omega_z+A_z).
\end{align}
Hence we define the Laplacian $\Delta_s^-$ as
\begin{equation}\label{DeltaS}
\Delta_s^-= \bar D^*\bar D, \quad\quad \Delta_s^-:\;
\mathcal C^{\infty}(\Sigma,L^k\otimes K^s)\to
\mathcal C^{\infty}(\Sigma,L^k\otimes K^s).
\end{equation}
We will also need to define the Laplacian with the operators
above interchanged
\begin{equation}\nonumber
\Delta_s^+= \bar D\bar D^*, \quad\quad \Delta_s^+:\;
\mathcal C^{\infty}(\Sigma,L^k\otimes K^s)\to
\mathcal C^{\infty}(\Sigma,L^k\otimes K^s).
\end{equation}
In particular, setting $s=k=0$  both Laplacians reduce to
the usual scalar Laplacian.
Since $\ker \bar D^*=0$, it follows from
Riemann-Roch theorem \eqref{RR1} that in our case
\begin{equation}
\nonumber
\dim\ker\Delta_s^-=\dim\ker\bar D=k+(1-{\rm g})(1-2s).
\end{equation}

The definition the regularized spectral determinants is standard.
Given non-zero eigenvalues ${\lambda}$ of $\Delta_{\rm L}$,
the regularized spectral determinant of the Laplacian is defined
as $\det'\Delta_{\rm L}=\exp(-\zeta'(0))$, where the zeta function
is $\zeta(u)=\sum_\lambda\lambda^{-u}$. In view of \eqref{barD},
the spectral determinants of both Laplacians \eqref{DeltaL}
and \eqref{DeltaS} are equal,
$\det'\Delta_s^-=\det'\Delta_{\rm L}$.
Now we would like to study the dependence of
the QH-state \eqref{slater} and the partition function $Z_k$
on the parameter space
$Y=\mathcal M_{\rm g}(\Sigma)\times \mathit{Jac}(\Sigma)$
\eqref{YY}.
Recall, that $\mathcal M_{\rm g}(\Sigma)$ is the moduli space 
of complex structures on $\Sigma$ and $\mathit{Jac}(\Sigma)$ is
the moduli space of flat connections $dA=0$  on
$\Sigma$ (Jacobian variety). Given the canonical basis
of one-cycles $(A_a,B_b)\in H_1(\Sigma,\mathbb Z),\;a,b=1,
\ldots,\rm g$  in $\Sigma$ and the dual basis of harmonic one-forms
$\alpha_a,\beta_b\in H^1(\Sigma,\mathbb Z)$, we can parameterize
the flat connections as follows,
\begin{equation}\label{flatcon}
A^{\rm flat}=2\pi\sum_{a=1}^{\rm g}(\varphi_1^a\alpha_a
-\varphi_2^a\beta_a).
\end{equation}
Here $(\varphi_1^a,\varphi_2^b)\in[0,1]^{2{\rm g}}$ are
coordinates on the Jacobian, which is $2{\rm g}$-dimensional
torus for the surfaces of genus ${\rm g}$.
This space corresponds to Aharonov-Bohm solenoid fluxes
piercing the handles of the surface.

The Slater determinant \eqref{slater} is a section $\mathcal S$
of the determinant line bundle $\mathcal L$ over the
parameter space, $\mathcal L=\det H^0(\Sigma, L^k\otimes K^s)$.
For the basis $\{s_j\}$ of
$H^0(\Sigma,L^k\otimes K^s)$ \eqref{basiswf},
the Quillen metric of  $\mathcal S$ of $\mathcal L$ is given by
\begin{equation}
\label{detsec}
\|\mathcal S\|^2
=\frac{\det \langle s_j,s_l\rangle}{\det'\Delta_{\rm L}}
=\frac{Z_k}{\det'\Delta_{\rm L}}\,,
\end{equation}
where the last equality follows from determinant formula
\eqref{detformula}.

\example The metric of the round sphere reads
$g_{0z\bz}=\frac1{(1+|z|^2)^2}$, where the complex coordinate
$z\in\mathbb C$, and the constant magnetic field $B=k$
corresponds to the hermitian metric on $L^k=\mathcal O(k)$
\begin{equation}\label{sphereherm}
h_0^k(z)=\frac1{(1+|z|^2)^k}.
\end{equation}
We can choose a basis of holomorphic sections in the form
$s_j(z)=c_jz^{j-1}(dz)^s,\,j=1,...,k-2s+1$, where $c_j$
is a normalization constant. The partition function then reads
\begin{align}\label{s2}
Z_{k}^{S^2}=\frac1{N_k!}\prod_j\frac{c_j^2}{2\pi}
\int_{\mathbb C^{N_k}}\prod_{j<l}|z_j-z_l|^2 \cdot
e^{-k\sum_j\psi(z_j)-s\sum_j\log\frac{\sqrt g}{\sqrt{g_0}}|_{z_j}}
\prod_{j=1}^{k-2s+1} \frac{d^2z_j}{(1+|z_j|^2)^{k-2s+2}},
\end{align}
where we used the formula $\det z_l^{j-1}=\prod_{j<l}(z_j-z_l)$
for the Vandermonde determinant. On the sphere there is no
complex structure moduli and $H^1(S^2,\mathbb Z)$ is trivial,
so the partition function is only a functional of the potential
functions $\phi,\psi$.

\example On the flat torus $T^2=\mathbb C/\Lambda,\;
\Lambda=m+n\tau,\,m,n\in\mathbb Z$ the metric reads
$g_{0z\bz}=\frac{2\pi i}{\tau-\bar\tau}\,\cdot$
We define a holomorphic line bundle $L\to T^2$ as follows.
Let $\Phi$ be the action of $\Lambda$ on $\C\times \C$
given by $\Phi(1)(z,\xi)=(z+1,\xi)$,
$\Phi(\tau)(z,\xi)=(z+\tau,e^{-2\pi i z -\pi i\tau}\xi)$.
We define $L=\C\times\C/\Lambda$, the quotient space
by this action, which is a holomorphic line bundle on 
$\C/\Lambda$ of degree one.
Let $L_\varphi=\C\times\C/\Lambda$ with the $\Lambda$  action
given by $\Phi_{\varphi}(1)(z,\xi)=(z+1,\xi)$,
$\Phi_{\varphi}(\tau)(z,\xi)=(z+\tau,
e^{- 2\pi i\varphi}\xi)$.
Constant magnetic field corresponds to the hermitian metric 
on $L^k\otimes L_\varphi$,
\begin{equation}\label{hermtorus}
h_0^k(z,\bz|\varphi,\bar\varphi)
=e^{\frac{\pi i k}{\tau-\bar\tau}(z-\bz)^2+\frac{2\pi i}
{\tau-\bar\tau}(z-\bz)(\varphi-\bar\varphi)},\:\:z\in\C.
\end{equation}
In this case we have a moduli parameter
$\tau\in\mathbb H/PSL(2,\mathbb Z)$ and also Jacobian variety
comes into play. We have two 1-cycles and correspondingly two
dual harmonic one-forms $\alpha^1,\beta^1$, and the gauge
connection \eqref{flatcon} has the form
\begin{equation}\label{flattorus}
A^{\rm flat}=2\pi(\varphi_1\alpha^1-\varphi_2\beta^1).
\end{equation}
The complex coordinate on the Jacobian is defined as
$\varphi=\varphi_2+\varphi_1\tau$. The basis of sections can
be written in terms of the theta functions with characteristics.
Also the canonical bundle on the torus is trivial, so we can set
$s=0$. Using standard notations for the theta functions \cite{M}
we can choose the basis of $H^0(\C/\Lambda,L^k\otimes L_\varphi)$ as
\begin{equation}\label{e:basis}
s_j(z)=\vartheta_{\frac jk,0}(kz+\varphi,k\tau),\,j=1,\ldots k.
\end{equation}
Using the $L^2$ norm of the section 
\begin{equation}\nonumber
\frac1{2\pi}\int_{T^2}\bs_j(\bz)s_l(z)h_0^k\sqrt{g_0}d^2z
=\sqrt{\frac i{k(\tau-\bar\tau)}}\cdot e^{-\frac{\pi i}{k}
\frac{(\varphi-\bar\varphi)^2}{\tau-\bar\tau}}\delta_{jl},
\end{equation}
we can find the partition function $Z_k$ for the flat torus and 
constant magnetic field
\begin{align}\label{torusZ}
Z_{k}^{T^2}(\tau,\bar\tau,\varphi,\bar\varphi)
=\frac{1}{k^{k/2}(2{\rm Im}\,\tau)^{k/2}}\cdot 
e^{-\pi i\frac{(\varphi-\bar\varphi)^2}{\tau-\bar\tau}}.
\end{align}
In this form the dependence of the partition function on
the parameter space coordinate $y=(\tau,\varphi)$ is explicit 
(recall that $Z_k$ is defined up to multiplication by mod squared 
of non-vanishing holomorphic function of $\tau$).

\section{Large $k$ asymptotics of the generating functional}
\label{sec3}

\subsection{Proof of Theorem \ref{prop1}}
\label{sec31}

The following is a more detailed formulation of the
Theorem \ref{prop1}. Here we write the formula Eq.\ \eqref{ano}
for the generation functional in terms of the \kahler and magnetic
potentials. This formulation is equivalent to Eq.\ \eqref{freeintro}
up to integration by parts. 
\begin{mythm}{\bf 1.}\label{thm1}
Consider the difference of generating functionals
$\log Z_k-\log Z_{k0}$ \eqref{Z1} for the metrics $h_0^k,h^k$ on
$L^k$ and $g_0,g$ on $\Sigma$, as defined in
Eqs.\ (\ref{kclass}, \ref{kclass1}).
The following asymptotic expansion holds at large $k$,
\begin{equation}\label{thm1exp}
\log \frac{Z_k}{Z_{k0}}=\log \frac{Z_H}{Z_{H0}}
+\mathcal F-\mathcal F_0,
\end{equation}
where the non-local (anomalous) part $\log Z_H$ is given by
\eqref{freeintro}. The difference $\log Z_H-\log Z_{H0}$ consists
of three leading orders in $k$
\begin{equation}\label{ano}
\log \frac{Z_H}{Z_{H0}}=-k^2S_2(\psi)
+k\frac{(1-2s)}2 S_1(\phi,\psi)
+\left(\frac{(1-2s)^2}4-\frac1{12}\right)S_L(\phi).
\end{equation}
where $S_2,\,S_1,\,S_L$ are geometric functionals. They read
\begin{align}\label{func1}
&S_2(\psi)=\frac1{2\pi}\int_\Sigma\left(\frac14\psi\Delta_0\psi
+\frac1k B_0\psi\right)\sqrt{g_0}d^2z,\\\label{func2}
&S_1(\phi,\psi)=\frac1{2\pi}\int_\Sigma\left(-\frac12\psi R_0
+\bigl(\frac1k B_0+\frac12\Delta_0\psi\bigr)
\log\bigl(1+\frac12\Delta_0\phi\bigr)\right)\sqrt{g_0}d^2z,\\
\label{func3}
&S_L(\phi)=\frac1{2\pi}\int_\Sigma\left(-\frac14\log\bigl(1
+\frac12\Delta_0\phi\bigr)\,\Delta_0\log\bigl(1
+\frac12\Delta_0\phi\bigr)+\frac12R_0\log\bigl(1
+\frac12\Delta_0\phi\bigr)\right)\sqrt{g_0}d^2z.
\end{align}
Here $\Delta_0$ is the scalar Laplacian in the reference metric
$g_0$ and $B_0=g_0^{z\bz}F_{0z\bz}$ is
the reference magnetic field.

The first terms of the expansion of the exact part $\mathcal F$
in \eqref{thm1exp} are given by
\begin{equation}\label{exact1}
\mathcal F=-\frac1{2\pi}\int_\Sigma\left[\frac12B\log \frac B{2\pi}
+\frac{2-3s}{12}R\log \frac B{2\pi}
+\frac1{24}(\log B)\Delta_g(\log B)\right]
\sqrt gd^2z+\mathcal O(1/k).
\end{equation}
\end{mythm}
\begin{proof}
We begin with the variational formula for the generating functional
$Z_k$, starting from the determinant formula \eqref{detformula}.
Denoting $G_{jl}=\langle s_j,s_l\rangle$, we get
\begin{align}\label{freeen}\nonumber
\delta \log Z_k&=\delta\,\tr\log \langle s_j,s_l\rangle
=-\frac1{2\pi}\sum_{j,l} G^{-1}_{lj}
\int_\Sigma\left(\frac{s-1}2(\Delta_g\delta\phi)
+k\delta\psi\right)\bs_js_lh^kg_{z\bz}^{-s}\sqrt gd^2z\\&
=-\frac1{2\pi}\int_\Sigma \left(\frac{s-1}2(\Delta_g B_k(z,\bz))\,
\delta\phi+kB_k(z,\bz)\,\delta\psi\right)\sqrt gd^2z,
\end{align}
where $B_k(z,\bz)$ is the Bergman kernel on
the diagonal for the line bundle $L^k\otimes K^s$.
Once we know the large $k$ asymptotic expansion of
the Bergman kernel on the diagonal we can determine the
asymptotic expansion of the generating functional.
The asymptotic expansion of the Bergman kernel is reviewed
in the Appendix B. There it is shown (see Eq.\ \eqref{eq:mm3.11})
that as $k\to\infty$,
\begin{equation}\label{expansion}
B_k(z,\bz)=B+\frac{1-2s}4R+\frac14\Delta_g\log B+
\frac{2-3s}{24}\Delta_g(B^{-1}R)+\frac1{24}\Delta_g(B^{-1}
\Delta_g\log B)
+\mathcal O(1/k^2)\,.
\end{equation}
Let us comment on the orders of different terms in this expansion.
The magnetic field $B$ is assumed to be of order $B\sim k\gg1$,
meaning $B(z,\bz)/k$ is a smooth positive function on $\Sigma$, 
independent of $k$. 
Therefore the first term here is of order $k$, the second and third
terms have order one and the last two terms have order $1/k$.
The inverse of $B$ is well-defined
due to the positivity condition $B>0$.
Note, that the expansion above is a gradient expansion,
in terms of derivatives
of the Riemannian metric and magnetic field.

Next we substitute the asymptotic expansion \eqref{expansion}
to  the variational formula \eqref{freeen} and rearrange
the terms as follows
\begin{align}\nonumber\label{term}
\delta \log Z_k=&-\frac1{2\pi}\int_\Sigma kB\,\delta\psi
\sqrt gd^2z+k\frac{1-2s}{8\pi}\int_\Sigma
\left(-R\,\delta\psi+\frac1k(\Delta_g B)\,\delta\phi\right)
\sqrt gd^2z\\\nonumber
&-\frac1{8\pi}\int_\Sigma\bigl(k(\Delta_g\log B)\,
\delta\psi-(\Delta_gB)\,\delta\phi\bigr)\sqrt gd^2z\\\nonumber&
+\left(\frac{(1-2s)^2}4-\frac1{12}\right)\frac1{8\pi}
\int_\Sigma (\Delta_g R)\,\delta \phi\sqrt gd^2z
\\\nonumber
&-\frac{2-3s}{48\pi}\int_\Sigma\left(k\Delta_g(B^{-1}R)\,
\delta\psi-(\Delta_g^2\log B)\,\delta\phi-(\Delta_gR)\,
\delta\phi\right)\sqrt gd^2z\\&
-\frac1{48\pi}\int_\Sigma\left( k\Delta_g(B^{-1}\Delta_g\log B)\,
\delta\psi-(\Delta_g^2\log B)\,\delta\phi\right)\sqrt gd^2z.
\end{align}
Writing out the definitions \eqref{def1}, \eqref{def2} explicitly,
\begin{equation}\nonumber
B=g^{z\bz}F_{z\bz}=\frac{B_0+\frac12k\Delta_0\psi}
{1+\frac12\Delta_0\phi},
\quad R=-\Delta_g\log\sqrt g
=\frac{R_0-\Delta_0\log(1+\frac12\Delta_0\phi)}
{1+\frac12\Delta_0\phi},
\end{equation}
(note that $B_0/k$ is independent of $k$) one can directly
verify the following variational formulas
for the geometric functionals \eqref{func1}-\eqref{func3},
\begin{align}\label{varfor1}
&\delta S_2=\frac1{2\pi }\int_\Sigma \frac1kB\delta\psi\,
\sqrt gd^2z,\\\label{varfor2}
&\delta S_1=\frac1{4\pi}
\int_\Sigma\left(-R\,\delta\psi+\frac1k(\Delta_g B)\,
\delta\phi\right)\sqrt
gd^2z,\\\label{varfor3}
&\delta S_L=\frac1{8\pi}\int_\Sigma (\Delta_g R)\,\delta\phi\,
\sqrt gd^2z,
\end{align}
and also
\begin{align}\label{BlogB}
&\delta\int_\Sigma B\log B\sqrt g d^2z
=\frac12\int_\Sigma\bigl(k(\Delta_g\log B)\,
\delta\psi-(\Delta_gB)\,\delta\phi\bigr)\sqrt gd^2z,\\\label{RlogB}
&\delta\int_\Sigma R\log B\sqrt g d^2z
=\frac12\int_\Sigma\left(k\bigl(\Delta_g(B^{-1}R)\bigr)\,
\delta\psi-(\Delta_g^2\log B)\,\delta\phi-(\Delta_gR)\,\delta\phi
\right)\sqrt gd^2z,\\\label{logBlogB}
&\delta\int_\Sigma(\log B)\Delta_g(\log B)\sqrt g d^2z
=\int_\Sigma\left(k\bigl(\Delta_g(B^{-1}\Delta_g\log B)\bigr)\,
\delta\psi-(\Delta_g^2\log B)\,\delta\phi\right)\sqrt gd^2z.
\end{align}
Now we can integrate Eq.\ \eqref{term} term by term.
In the first line of \eqref{term} we recognize the variations
of the functionals $S_2$ and $S_1$,
in the second line -- the variation of \eqref{BlogB},
in the third line -- the variation of $S_L$ \eqref{func3},
and the variations of \eqref{RlogB}, \eqref{logBlogB} in the last two
lines of \eqref{term}.

When integrating this formulas, we need to take into
account the obvious boundary condition: at $\phi=\psi=0$
we have $\log Z_k-\log Z_{k0}=0$. This completes the proof.
\end{proof}

\begin{rem}
Using Theorem \ref{prop11}, we interpret $\mathcal{F}$ as
$\log \det'\Delta_{\rm L}$. While we computed  $\mathcal F$
to the order $\mathcal O(1)$, in fact that all $1/k$-corrections
to $\log \det'\Delta_{\rm L}$
are given by local functionals. This follows from the general
arguments given in \cite[\S 5.5.4]{MM}. We also note that physics
proof of the locality
of $\mathcal F$ was also given in Ref.\ \cite[\S 5]{FK}.

Integrating the left hand side of \eqref{expansion} over
the surface, we should get the total number of LLL states. We have
\begin{equation}\nonumber
N_k=\frac1{2\pi}\int_\Sigma B_k(z,\bar z)\sqrt gd^2z
=k+(1-{\rm g})(1-2s),
\end{equation}
in accordance with \eqref{RR1}. Indeed, only the first two terms
on the right hand side in \eqref{expansion} contribute to
the integral, while the rest of the expansion  is a 
total derivative.
Let us also mention that the Bergman kernel is the average density
of states
\begin{equation}\nonumber
\langle\rho(z,\bz)\rangle=\frac1{2\pi}B_k(z,\bz),
\end{equation}
as follows from \eqref{freeen} and \eqref{dens}.

In the case of constant magnetic field $B=k$ and $s=0$
the expansion \eqref{ano}  is in
agreement\footnote{The Riemannian conventions for
the scalar curvature and Laplacian used here are different by
a factor 2 from \kahler conventions used in \cite{K}.
Then the expansion of the Bergman kernel in \cite{K}
follows from \eqref{expansion} after setting the magnetic field
to constant: $B=k$.} with the one obtained in \cite{K}.
In particular, functionals $S_1$ \eqref{func1} and
$S_2$ \eqref{func2} reduce to the combinations of the
Aubin-Yau and Mabuchi functionals in \kahler geometry.
The functional $S_L$ \eqref{func3} is the Liouville functional.
The exact terms in Eq.\ \eqref{exact1} are in agreement
with \cite[Eq.\ (130)]{CLW1}.

 As we have mentioned, the formula Eq.\ \eqref{freeintro} for
 the anomalous part $\log Z_H$ of the generating functional
 \begin{align}\nonumber
\log \frac{Z_H}{Z_{H0}}= &\frac2{\pi}\int_\Sigma\left[A_zA_{\bz}
+\frac{1-2s}2(A_z\omega_{\bz}+
\omega_zA_{\bz})+\left(\frac{(1-2s)^2}4-\frac1{12}\right)
\omega_z\omega_{\bz}\right]d^2z\\
&-(A\to A_0, \omega\to\omega_0)
\end{align}
follows from
\eqref{ano} by integration by parts. 
This formula is understood to be valid specifically in 
the symmetric gauge, where the coefficients of the 
gauge connection and spin connection 1-forms 
(\ref{gauge2}), (\ref{gauge1}).
This is a somewhat 
formal, but nevertheless useful way to represent 
Eq.\ \eqref{freeintro}, 
because it resembles Chern-Simons functional, which
will appear in \S \ref{CSform}.

We emphasize that
the anomalous part is a non-local functional of magnetic field
and the curvature, while the exact terms
are {\it local}, i.e., involve integrals of the
magnetic field density and scalar curvature and their derivatives
at a point. This can be illustrated
by another non-local representation for the anomalous
term via double integrals
\begin{align}\nonumber
\log \frac{Z_H}{Z_{H0}}=&-\frac1{2\pi} \int_{\Sigma\times\Sigma}
\left(B+\frac{1-2s}4R\right)|_z
\Delta_g^{-1}(z,z')\left(B+\frac{1-2s}4R\right)|_{z'}
\sqrt gd^2z\sqrt gd^2z'\\&\nonumber
+\frac1{96\pi}\int_{\Sigma\times\Sigma}  
R|_z\Delta_g^{-1}(z,z')R|_{z'}\sqrt gd^2z\sqrt gd^2z' 
-(R\to R_0,B\to B_0),
\end{align}
where $\Delta_g^{-1}$ is the inverse Laplace operator.
This formula is an analog of the
Polyakov effective action in $2d$ gravity \cite{P1},
in the presence of the magnetic field $B$.
\end{rem}

\section{Regularized spectral determinant of the Laplacian}
\label{sec4}

In \S \ref{Y} we have seen that the norm squared
$\|\mathcal S\|^2$ of a section of the determinant line bundle
$\mathcal L$ on $Y$ is given by the ratio of the
generating functional $Z_k$ and regularized determinant of the
Laplacian for the line bundle \eqref{detsec}. Here we show that
the anomalous part of the generating functional \eqref{ano} is
encoded in the Quillen
metric \eqref{detsec}. The following is a more detailed formulation
of Theorem \ref{prop11}.
We will use the small time expansion
of the heat kernel taking advantage of the fact that this expansion
is easy to handle on Riemann surfaces. Thus we give a direct proof here,
avoiding the sophisticated local index arguments from
\cite[Theorem 1.23]{BGS}.
\begin{thm}
\label{thm2}
Consider the metrics $h_0^k,h^k$ on $L^k$ and $g_0,g$ on
$\Sigma$ related as in
Eqs.\ \eqref{kclass}, \eqref{kclass1}. The norm of the
determinant section $\|\mathcal S\|^2$
is defined  as in Eq.\ \eqref{detsec} for the metrics $g,h^k$,
and $\|\mathcal S\|_0^2$
for the metrics $g_0,h_0^k$. Then the following exact
transformation formula holds
\begin{align}\label{thm21}
\log\frac{\|\mathcal S\|^2}{\|\mathcal S\|^2_0}=-k^2S_2(\psi)
+k\frac{(1-2s)}2  S_1(\phi,\psi)+\left(\frac{(1-2s)^2}4
-\frac1{12}\right)S_L(\phi),
\end{align}
where the functionals are defined in
Eqs.\ \eqref{func1}-\eqref{func3}.
\end{thm}
\begin{proof}

The derivation here is based on the standard heat kernel
technique, (see e.\ g.\  \cite[\S  IIF]{DP},
where the canonical line bundle $K^s$ was considered.)
{We begin} with the heat kernel representation of
the regularized determinant
\begin{equation}\nonumber
\log{\det}' \Delta_s^-=-\int_\epsilon^\infty\frac{dt}t
\bigl(\tr e^{-t\Delta_s^-}-N_k\bigr),
\end{equation}
where $\epsilon\to 0$. Now we compute the variation with
respect to the \kahler potential $\delta\phi$ and magnetic
potential $\delta\psi$
\begin{equation}
\label{var}
\delta\log{\det}' \Delta_s^-=\int_\epsilon^\infty\,dt\,
\tr \delta \Delta_s^-\,e^{-t\Delta_s^-}.
\end{equation}
The variations of the first-order operators are given by
\begin{align}\nonumber
&\delta\bar D=-\frac14\bigl[(s+1)\Delta_g\delta\phi
+ 2k\delta\psi\bigr]\bar D+\frac14\bar D
\bigl[s\Delta_g\delta\phi
+ 2k\delta\psi\bigr],\\
&\delta\bar D^*=\frac14\bigl[(s-2)\Delta_g\delta\phi
+ 2k\delta\psi\bigr]\bar D^*-\frac14\bar D^*\bigl[(s-1)
\Delta_g\delta\phi+ 2k\delta\psi\bigr].
\end{align}
Using the formulas above we derive
\begin{align}\nonumber
\tr \delta \Delta_s^-\,e^{-t\Delta_s^-}
=&\tr \bigl(\delta\bar D^* \bar D+\bar D^*\delta\bar D\bigr)\,
e^{-t\Delta_s^-}\\
\nonumber=&\tr \left(\frac14\bigl[(s-2)\Delta_g\delta\phi
+ 2k\delta\psi\bigr]\Delta_s^--\frac14\bar D^*\bigl[(s-1)
\Delta_g\delta\phi+ 2k\delta\psi\bigr]\bar D\right.\\\nonumber
&\left.-\frac14\bar D^*\bigl[(s+1)\Delta_g\delta\phi
+ 2k\delta\psi\bigr]\bar D
+\frac14\Delta_s^-\bigl[s\Delta_g\delta\phi+ 2k\delta\psi\bigr]
\right)\,e^{-t\Delta_s^-}\\
\nonumber=&\tr \left(\frac{s-1}2\Delta_g\delta\phi
+k\delta\psi\right)\Delta_s^-\,e^{-t\Delta_s^-}-
\tr\left(\frac s2\Delta_g\delta\phi+k\delta\psi\right)
\Delta_{s-1}^+\,e^{-t\Delta_{s-1}^+}.
\end{align}
Here in the middle terms we used the rearrangement identity
$\tr\bar D\,e^{-t\Delta_s^-}\bar D^*
=\tr\Delta_{s-1}^+e^{-t\Delta_{s-1}^+}$.
Going back to \eqref{var} and performing $t$-integral we obtain
\begin{equation}
\label{var11}
\delta\log{\det}' \Delta_s^-
=-\tr\left(\frac{s-1}2\Delta_g\delta\phi+k\delta\psi\right)\,
e^{-t\Delta_s^-}|_\epsilon^\infty
+\tr\left(\frac s2\Delta_g\delta\phi+k\delta\psi\right)\,
e^{-t\Delta_{s-1}^+}|_\epsilon^\infty.
\end{equation}
As $t\to\infty$ the first heat kernel above is projected onto
its zero modes \eqref{basiswf} and the second one vanishes,
since $\ker \Delta_s^+=0$. The $t=\infty$ part reduces to
\begin{align}
\label{free2}\nonumber
-\tr\left(\frac{s-1}2\Delta_g\delta\phi+k\delta\psi\right)\,
e^{-t\Delta_s^-}|_{t=\infty}
&=-\sum_{j=1}^{N_k}\left\langle s_j,\left(\frac{s-1}2\Delta_g
\delta\phi+k\delta\psi\right)s_j\right\rangle\\
&=-\delta\log\det\langle s_j,s_l\rangle,
\end{align}
where we recognize the variation formula \eqref{freeen}
for the generating functional in the determinantal representation
\eqref{detformula}. Using \eqref{free2}, Eq.\ \eqref{var11}
can be now written as
\begin{align}
\label{var1}
\delta\log\frac{\det\Delta_s^{-}}{Z_k}
=\tr \left(\frac{s-1}2\Delta_g\delta\phi+k\delta\psi\right)\,
e^{-\epsilon\Delta_s^-}\big|_{\epsilon\to0}
-\tr\left(\frac s2\Delta_g\delta\phi+k\delta\psi\right)\,
e^{-\epsilon\Delta_{s-1}^+}\big|_{\epsilon\to0}\,.
\end{align}
We evaluate these terms with the help of the short-time expansion
formulas for the heat kernel on the diagonal
\begin{align}
\label{heat}
&\langle z|e^{-\epsilon\Delta_s^-}|z\rangle=
\frac1{2\pi \epsilon}+\frac{B}{4\pi}+\frac{1-3s}{24\pi}R
+\mathcal O(\epsilon),\\
\nonumber&\langle z|e^{-\epsilon\Delta_{s-1}^+}|z\rangle
=\frac1{2\pi \epsilon}-\frac{B}{4\pi}+\frac{-2+3s}{24\pi}R
+\mathcal O(\epsilon)\,.
\end{align}
Plugging this in Eq.\ \eqref{var1} we note that the singular
$1/\epsilon$ terms cancel and we recognize the variational
formulas \eqref{varfor1}-\eqref{varfor3} for the geometric
functionals,
\begin{align}\nonumber
\delta\log\frac{\det'\Delta_s^{-}}{Z_k}=k^2\delta S_2(\psi)
-k\frac{(1-2s)}2 \delta S_1(\phi,\psi)-\left(\frac{(1-2s)^2}4
-\frac1{12}\right)\delta S_L(\phi).
\end{align}
Integrating this variational formula with the boundary condition
$\|\mathcal S\|^2=\|\mathcal S\|_0^2$ at $\phi=\psi=0$ completes
the proof.
\end{proof}

As we have already pointed out, this is a special case of
a more general anomaly formula
for the Quillen metric, \cite[Theorem 1.23]{BGS}, 
which holds for any K\"ahler manifold, see
also \cite[Proposition 3.8]{Fay}, see Appendix A for more details.
Note also that the formula Eq.\ \eqref{thm21} is exact and
the magnetic field
was not assumed to be large in the derivation of this result.

The following corollary follows from Theorems \ref{prop1}
and \ref{prop11}.
\begin{mycor} The regularized spectral determinant of the Laplacian
\eqref{DeltaS} satisfies the following transformation formula
\begin{align}\label{bv}
\log\,&\frac{\det'\Delta_{\rm L}}{\det'\Delta_{{\rm L}0}}
=\mathcal{F}-\mathcal{F}_0,
\end{align}
where we used the equality
$\det'\Delta_{\rm L}=\det'\Delta_{s}^{-}$ for the Laplacian
in Eq.\ \eqref{DeltaS} and \(\mathcal{F}\)
is given by Eq.\ \eqref{exact1}.\end{mycor}

The first term $\int_\Sigma B\log B$ in \eqref{exact1}  agrees
with the result of
Bismut-Vasserot \cite{BV} for the leading asymptotics of
the analytic torsion for higher powers
of line bundles on \kahler manifolds,
see also \cite[Theorem 5.5.8\,]{MM}. The last two terms
in \eqref{exact1}  are
$\mathcal O(1)$ corrections to their result.
We conjecture that the coefficient of $k^0$ in $\mathcal{F}$ from
\eqref{freeen} is the  coefficient of $k^0$ in the expansion for
$\log \det'\Delta_{\rm L}$, which would refine
Bismut-Vasserot's result.

\example On the round sphere with the line bundle
$L^k=\CO(k)$ with the  hermitian metric
\eqref{sphereherm}, the regularized determinant of  Laplacian
is given by
\begin{align}\label{apb5.4}
 \log {\det}'\Delta_{L^k}
 =2\sum_{j=1}^{k} (k-j) \log (j+1)-(k+1)\log(k+1)!
 -4\zeta'(-1)+\frac{(k+1)^2}2,
\end{align}
see \cite[Theorem 18]{Ko95}
\footnote{Cf. also \cite[\S 5.1]{Weng95},  with the caveat that 
the metric \eqref{sphereherm} is used incorrectly 
and the final result at p.\ 352 needs to be corrected.}.
 By the Euler-Maclaurin formula, at large $k$ the equation above
 reduces to
 \begin{align}\label{apb5.4a}
 \log {\det}'\Delta_{L^k}=-\frac k2\log\frac k{2\pi}-\frac23\log k
 +\mathcal O(1),
 \end{align}
 which is in accordance with Eq.\ \eqref{bv} (the scalar curvature of
 round sphere is $R_0=4$), where the exact part is given by
 Eq.\ \eqref{exact1}.

\example On the flat torus of area $2\pi$ with the line bundle
$L^k\otimes L_\varphi$ of degree $k>0$, endowed with 
its canonical translation
invariant curvature metric \eqref{hermtorus}, the regularized
determinant of Laplacian equals
(cf. \cite[Prop. 4.2]{Bo96},
\cite[pp. 4469]{Ber01}),
\begin{align}\label{apb5.5}
\log {\det}'\Delta_{L^k\otimes L_\varphi}
= - \frac{k}{2} \log \frac{k}{2\pi}\,\cdot
\end{align}
This is actually an exact formula and the reminder term $o(k)$
vanishes. Moreover, the result does not depend on the flat line
bundle \eqref{flattorus}. Combining this with Eq.\ \eqref{torusZ}, 
the norm of the section in the Quillen metric \eqref{detsec} reads
\begin{equation}
\|\mathcal S\|^2=\frac{1}{(4\pi{\rm Im}\,\tau)^{k/2}}
\cdot e^{-\pi i\frac{(\varphi-\bar\varphi)^2}{\tau-\bar\tau}}
\end{equation}
In particular, this result implies that the adiabatic curvature 
$\Omega$ and the Chern curvature $\Omega^{\mathcal L}$ 
of the Quillen metric  
\eqref{omega2} coincide and equal to
\begin{equation}\nonumber
\Omega=\Omega^{\mathcal L}
=\frac{2\pi i}{\tau-\bar\tau}id\varphi\wedge d\bar\varphi
-\frac k{2(\tau-\bar\tau)^2}id\tau\wedge d\bar\tau.
\end{equation}
This is exactly the formula derived in Ref.\ \cite{ASZ1}.
The first term here is a $(1,1)$-form on $Jac(\Sigma)$ and describes 
Hall conductance, as we explained in \S \ref{PIntr}, 
and the second term is a $(1,1)$-form on $\mathcal M_1$ related 
to Hall viscosity.

\example On higher genus surfaces the determinant of laplacian
for the metric of constant scalar curvature and canonical line bundle 
has the form, see e.g.\ \cite[Eq.\ 15]{ASZ} and \cite{DP1},
\begin{equation}\nonumber
{\det}'\Delta_{L^k}=e^{-c_k\chi(\Sigma)}
\prod_{\gamma}\prod_{j=1}^\infty
[1-e^{i\sum_{a=1}^{2{\rm g}}\varphi_an_a(\gamma)}
e^{-(j+k)l(\gamma)}].
\end{equation} 
 This is written for 
the surfaces realized as orbit spaces of discrete subgroups
$\Gamma$ of $SL(2,\mathbb R)$ acting on upper half plane. 
Here $\gamma\in\Gamma$ are primitive 
hyperbolic elements of $\Gamma$ representing conjugacy 
classes corresponding to closed geodesics on $\Sigma$,  
$n_a$ counts the number of times the closed geodesic  
goes around the $a$th fundamental loop, and 
$l(\gamma)$ is the length of geodesic. 
Also  $c_k$ is a constant, see \cite{DP1} for details. 
The leading term in the asymptotic of 
${\det}'\Delta_{L^k}$ is given by $e^{-kl(\gamma_{\rm min})}$, 
where $\gamma_{\rm min}$ is the length of shortest geodesics. 
Therefore when the latter is non-vanishing, i.e., away from the
boundary of the moduli space this term represents small fluctuations
and $\Omega^{\mathcal L}$ approximates $\Omega$ 
with exponential precision.

\section{Curvature formula for Quillen metric and adiabatic transport}
\label{sec5}
\subsection{Curvature formula and adiabatic
transport in QHE}\label{CSform}

In this section we discuss the adiabatic curvature of the integer
quantum Hall state and its relation to the Quillen metric. 
Here we consider the family of Riemann surfaces $\Sigma_y$,
parameterized by $y\in Y$. Let $M$ be the space, which
is the union of all $\Sigma_y$ over $Y$ (universal curve),
and denote by $\sigma:M\to Y$ the natural projection.
Consider also the family of line bundles
$L^k_y\otimes K_y^s\to\Sigma_y$, $y\in Y$.
This extends to the holomorphic line bundle $E$ over $M$,
which is the union of all line bundles $L^k_y\otimes K_y^s$.
The Hermitian metric $h^k(z,\bz,y,\bar y)$ straightforwardly 
extends to the Hermitian
metric $h^E$ on $E$ over $M$. 

\begin{mainlem}
The adiabatic connection equals the Chern connection
of the $L^2$-metric on $\mathcal{L}$. Hence,
the adiabatic curvature $\Omega$ is the curvature of 
the Chern connection
of the $L^2$-metric on $\mathcal{L}$.
\end{mainlem}
\begin{proof}
As we explained in \S \ref{PIntr}, in quantum mechanics 
the adiabatic connection is defined as follows \cite{Simon}. 
Let $\Psi(t)$ be a normalized wave function 
$\langle\Psi(t),\Psi(t)\rangle_{L^2}=1$, belonging to some
(multi-dimensional) parameter space $t\in Y$.
Then the adiabatic (Berry) connection is given by the formula 
Eq.\ \eqref{aconn},
$$\mathcal A_t=i\langle\Psi(t),d_t\Psi(t)\rangle_{L^2}.$$ 
In the setup of the present paper we need a slightly 
more general version of this connection, due to the following. 
Our wave function $\mathcal S(y)$ Eq.\ \eqref{Phi1} depends 
on the coordinates $t=(y,\bar y)$ on the parameter space $Y$ 
holomorphically and, in fact, it is a section of the line bundle 
$\mathcal L$ on $Y$ as well as ${\rm L}$ on $\Sigma$
(for each coordinate $z_j$). Therefore the derivative $d_t$ 
in the definition above in our context is replaced by the appropriate 
covariant derivative $\nabla^E_{u^H}$, which is 
the Chern connection on $E$, see Eq.\ \eqref{eq5:}-\eqref{eq8:}.
Introducing the notation 
$\Psi(y,\bar y)=\mathcal S(y)/\|\mathcal S(y)\|_{L^2}$, 
where $\Psi(y,\bar y)$ is the normalized wave function  
$\langle\Psi(y,\bar y),\Psi(y,\bar y)\rangle_{L^2}=1$, 
we shall define the adiabatic connection $\mathcal A_y$
in this general setup by 
\begin{equation}\label{e:adiabcon}
\mathcal A_y(u)=i\langle \Psi (y,\bar y),\nabla^E_{u^H}
\Psi(y,\bar y)\rangle_{L^2}, \quad {\rm for\; any}\;\; u\in 
T^{(1,0)}Y,
\end{equation} 
where $u^H$ is horizontal lift of $u$. 
Now we would like to show that
\begin{equation}\label{e:adiabcon1}
\mathcal A_y=
\frac{i}{2}\p_y\log\|\mathcal S(y)\|^2_{L^2}.
\end{equation} 
We have 
\begin{align}
&\mathcal A_y(u)=i\langle \mathcal S(y)/\|\mathcal S(y)\|_{L^2}, 
\nabla^E_{u^H} 
\bigl(\mathcal S(y)/\|\mathcal S(y)\|_{L^2}\bigr)\rangle_{L^2}\\
&=i\langle \mathcal S(y)/\|\mathcal S(y)\|_{L^2}, 
\bigl(\nabla^E_{u^H} \mathcal S(y)\bigr)/\|\mathcal S(y)\|_{L^2}
\rangle_{L^2}
-\frac i2i_u\p_y\log \|\mathcal S(y)\|^2_{L^2}\\
&=i\frac1{\|\mathcal S(y)\|_{L^2}}\langle \mathcal S(y), 
(\nabla^E_{u^H} \mathcal S(y))
\rangle_{L^2}-\frac i2i_u\p_y\log \|\mathcal S(y)\|^2_{L^2},
\end{align} 
where $i_u$ is contraction with the vector $u$. 
Then the first term in the last line equals
\begin{align}
\langle \mathcal S(y), (\nabla^E_{u^H} \mathcal S(y))\rangle_{L^2}
=i_u\p_y\langle \mathcal S(y), \mathcal S(y)\rangle_{L^2}-\langle 
\nabla^E_{{\bar u}^H} \mathcal S(y), \mathcal S(y)\rangle_{L^2},
\end{align} 
where the last term vanishes because the section is holomorphic. 
Here we crucially used the fact that Lie derivative in direction 
$u^H$ of the volume form $L_{u^H}\omega_{\Sigma_y}=0$, 
as explained in Appendix A, Eq.\ \eqref{lie}. 
Then the statement of Eq.\ \eqref{e:adiabcon1} follows. 
We can also write the connection Eq.\ \eqref{e:adiabcon1}
in the form
\begin{equation}\label{e:adiabcon2}
\mathcal A_y=\frac{i}{2}\p_y\log Z_k.
\end{equation} 
The proof is complete.
\end{proof}
From the Eq.\ \eqref{e:adiabcon2} above we arrive at the formula, 
which relates adiabatic curvature
to the generating functional, 
\begin{equation}\label{adcurv1}
\Omega=d_Y\mathcal A=-\partial_Y \overline\partial_{Y}\log Z_k\,.
\end{equation}
The adiabatic curvature
$\Omega$ is a closed $(1,1)$-form on the parameter space $Y$.
The integrals of $\Omega$ over closed 2-cycles in $Y$ \eqref{YY}
are the adiabatic transport coefficients. 
The crucial physical observation,
that goes back to \cite{T1,T2,AS}, attributes the quantization
of the Hall conductance to the fact that
$\frac i{2\pi}\Omega$ belongs to an integral cohomology class
if $Y=\mathit{Jac}(\Sigma)$.
Another observation due to \cite{ASZ,L1,L2,KW} is that for
$Y=\mathcal M_{\rm g}$
the curvature $\frac i{2\pi}\Omega$ belongs to
a rational cohomology
class and defines geometric transport on the moduli space of 
complex structures.

Note that following formula holds \begin{equation}\label{adcurv}
\Omega=\Omega^{\mathcal L}
-\partial_Y \overline\partial_{Y}\log {\det}'\Delta_{\rm L},
\end{equation}
where $\Omega^{\mathcal L}
=\partial_Y \overline\partial_{Y}\log \|\mathcal S\|^2$
is the curvature of Chern connection of the Quillen metric 
\eqref{detsec} on the determinant line bundle $\mathcal L$.
The last term in \eqref{adcurv}, is in fact an exact form on $Y$,
since ${\det}'\Delta_{\rm L}$ is a function on $Y$.
It vanishes under the integration
over a closed 2-cycle, in which case the only contribution to the
transport coefficients comes from  $\Omega^{\mathcal L}$.

Now we recall the curvature formula,
following \cite{BGS}, see also \cite{BF,BJ}.
We also note that this formula appeared in the QHE context
in Ref.\ \cite{TP06}.  
We denote the curvature of the metric $h^E$ on $E$ over $M$ 
by $\mathrm F^E$. In the same way we can extend the tangent
bundle $T\Sigma$ to the union $TM|Y$ (note that this bundle
is one-dimensional while $TM$ is not) over $M$, and let
$g^{TM|Y}$ be any smooth Hermitian metric on $TM|Y$
(which is automatically K\"ahler along fibers since they are
one-dimensional) with $R_{TM|Y}$ being its curvature 2-form.

Then the following formula for the curvature of
the determinant line bundle $\Omega^{\mathcal L}$
holds \cite[Theorems \,1.9, 1.27]{BGS}\footnote{
In the literature, there exist two different curvature formulas.
In the smooth category, for a family of
Dirac operators, Bismut-Freed defined in \cite[Theorem\,1.21]{BF}
a unitary connection on the smooth determinant line bundle and
computed its curvature.
In the holomorphic category which we use here,
Bismut-Gillet-Soul\'e \cite[Theorem 1.27]{BGS} computed
the curvature of the Chern connection of a holomorphic
determinant line bundle
for a holomorphic locally K\"ahler fibration
\cite[Definition\,1.25]{BGS} with compact fiber.
All curvatures in \eqref{GRR} are curvatures of associated
Chern connections, and
thus the integral of the right hand side of \eqref{GRR} is
a $(1,1)$-form on the base manifold.}
\begin{equation}\label{GRR}
\Omega^{\mathcal L}
=-2\pi i\int_{M|Y}\bigl[{\rm Ch}(E){\rm Td}(TM|Y)\bigr]_{(4)}.
\end{equation}
Here the integrand is a form of mixed degree on $M$.
The subscript $4$ means that only $4$-form component
of the full expression is retained, so that the result of the
integration is a 2-from. The notation $M|Y$ means that
the integration goes over the fibers in the fibration
$\sigma:M\to Y$, i.e., over the spaces $\Sigma_y$ at  $y$ fixed.

In order to apply \cite[Theorem\,1.27]{BGS} we need to check
that our fibration $\sigma: M\to Y$ with Riemann surfaces
is locally K\"ahler. This can be seen as follows.
Let $L$ be a holomorphic line bundle on $M$ which is positive along
each fiber $\Sigma_{y}$ (for example, we can take
$L|_{\Sigma_{y}}= K_{\Sigma_{y}}$
if the genus of  $\Sigma$ satisfies $\mathrm{g}>1$).
By taking a sufficiently large power of $L$, we can assume that
$\deg (L|_{\Sigma_{y}})> 2(1-\mathrm{g})+2$ for every $y\in Y$.
Then by \cite[p.\,215]{GH:78}, the Kodaira map
$\Sigma_{y}\to \mathbb{P}(H^0(\Sigma_{y}, L|_{\Sigma_{y}})^*)$
along the fiber $\Sigma_{y}$
is a holomorphic embedding for each $y\in Y$. These Kodaira maps
vary holomorphically with $y\in Y$ and induce
a holomorphic embedding of $M$ onto the total space of
the holomorphic vector bundle
$\mathbb{P}(H^0(\Sigma, L|_{\Sigma})^*)$ over $Y$,
whose fibers are
$\mathbb{P}(H^0(\Sigma_{y}, L|_{\Sigma_{y}})^*)$.
A smooth metric on the  holomorphic vector bundle
$H^0(\Sigma, L|_{\Sigma})$ on $Y$, it induces a smooth metric
$h^{\mathcal{O}(1)}$ on
the fiberwise hyperplane line bundle $\mathcal{O}(1)$
on $\mathbb{P}(H^0(\Sigma, L|_{\Sigma})^*)$.
Now the restriction of the first Chern form of the Chern connection
of $(\mathcal{O}(1), h^{\mathcal{O}(1)})$
on $M$ gives a closed $(1,1)$-form which is K\"ahler
along each fiber $\Sigma_{y}$. This shows that
$\sigma: M\to Y$ is actually a 
K\"ahler fibration \cite[Definition\,1.4]{BGSII}, 
cf.\ Appendix A.

\subsection{Proof of Theorem \ref{prop2}}
\label{sec52}
Let us write $\Omega^{\mathcal L}$ in Eq.\ \eqref{GRR}
locally as an exterior derivative of a 1-form on
$Y$: $\Omega^{\mathcal L}= d_Y\mathcal A^{\mathcal L}$.
Then we choose an adiabatic process,
i.e., a smooth closed contour $\mathcal C$ in $Y$.
When the system is transported along the contour, the geometric
part of the  adiabatic phase (associated with Quillen metric)
is given by the integral of the connection along $\mathcal C$
\begin{equation}\label{adphase}
\int_{\mathcal C}\mathcal A^{\mathcal L}.
\end{equation}
Now we would like to obtain an explicit formula for this integral.
The goal is to show that it is given by (2+1)d integral of
a specific Chern-Simons form over $\sigma^{-1}(\mathcal{C})$.

Let us focus on the structure of the integrand in \eqref{GRR}.
The Chern character form ${\rm Ch}(E)$
of a vector bundle $E\to M$ and the Todd form ${\rm Td}(TM)$
of the tangent bundle of a manifold
$M$ are defined by the formal expansions in the powers of the
corresponding curvature forms.
Here we will need only a few first terms in these expansions
\begin{align}\nonumber
&{\rm Ch}(E)=1+c_1(E)+\frac12\bigl(c_1^2(E)-2c_2(E)\bigr)
+\ldots,\\
&{\rm Td}(TM)=1+\frac12c_1(TM)+\frac1{12}\bigl(c_1^2(TM)
+c_2(TM)\bigr)+\ldots
\end{align}
where the Chern forms are given by
\begin{align}\nonumber
&c_1(E)=\frac{i}{2\pi}\tr  \mathrm{F}^E,\quad c_2(E)=
\frac1{8\pi^2}\bigl(\tr \mathrm{F}^E\wedge
\mathrm{F}^E-\tr  \mathrm{F}^E\wedge \tr \mathrm{F}^E\bigr),\\
&c_1(TM)=\frac{i}{2\pi}\tr \mathrm{R}^{TM},\quad c_2(TM)
=\frac1{8\pi^2}\bigl(\tr \mathrm{R}^{TM}\wedge \mathrm{R}^{TM}
-\tr \mathrm{R}^{TM}\wedge \tr \mathrm{R}^{TM}\bigr).
\end{align}
In our case the 2-forms $\mathrm{F}^E$ and $\mathrm{R}_{TM|Y}$
are scalar valued,
so the traces shall be omitted, hence $c_2(E)=c_2(TM|Y)=0$.
Also we split the curvature 2-form of the bundle $E$ as:
$\mathrm{F}^{E}=\mathrm{F}-s\mathrm{R}_{TM|Y}$
where $\mathrm{F}$ now refers to the part of the curvature
2-form corresponding to the line bundle
$\tilde L^k\to M$, which is the union of all bundles
$L^k_y\to \Sigma_y$.
Using the composition property of the Chern character forms
${\rm Ch}(E\otimes E')={\rm Ch}(E)\cdot {\rm Ch}(E')$
for any two bundles $E,E'$, we get
\begin{align}
\nonumber
{\rm Ch}(L^k\otimes K^s)
=&1+c_1(L^k)-sc_1(TM|Y)-sc_1(L^k)c_1(TM|Y)
+\frac12\bigl(c_1^2(L^k)+s^2c_1^2(TM|Y)\bigr)+\ldots
\end{align}
Then the formula \eqref{GRR} specified to our case reads
\begin{align}\label{BF}
\Omega^{\mathcal L}=\frac{i}{4\pi}
\int_{M|Y}\left[\mathrm{F}\wedge \mathrm{F}
+(1-2s)\, \mathrm{F}\wedge \mathrm{R}_{TM|Y}
+\left(\frac{(1-2s)^2}4-\frac1{12}\right)\,
\mathrm{R}_{TM|Y}\wedge \mathrm{R}_{TM|Y}\right]&.
\end{align}
The holonomy
    $\exp(- {\int_{\mathcal C}\mathcal A^{\mathcal L}})$
    is the parallel transport operator from the point $\mathcal{C}_{0}$
    into  $\mathcal{C}_{1}(=\mathcal{C}_{0})$
    along the path $\mathcal{C}_{s}$ for the Chern
    connection on $(\mathcal{L}, \|\cdot\|)$. 
    It is a complex number of norm $1$
    and it does not depend on the origin $\mathcal{C}_{0}$.
 If $\mathcal{C}$ is the boundary of a surface $G$,
    then by (\ref{GRR}),
    \begin{align}\label{eq:gCS51}
\exp\Big(- {\int_{\mathcal C}\mathcal A^{\mathcal L}}\Big)
    = \exp\Big(2i\pi\int_{\sigma^{-1}(G)}
 \bigl[{\rm Ch}(E){\rm Td}(TM|Y)\bigr]_{(4)}\Big).
\end{align}   
    
   By \cite[Theorem 1.15]{BGS}, the unitary connection
  in \cite{BF} associated with the horizontal 
  bundle (the orthogonal bundle to the vertical bundle)
  is exactly the Chern connection
  on  $(\mathcal{L}, \|\cdot\|)$, thus by \cite[Theorem 3.16]{BF}, 
 \begin{align}\label{eq:gCS61}
\exp\Big(\!- {\int_{\mathcal C}\mathcal A^{\mathcal L}}\Big)
    = \exp\big(2i\pi \overline{\eta}\big),
\end{align}
where $\overline{\eta}$ is the adiabatic limit of the
reduced  eta invariant for the fibration 
    $\sigma: \sigma^{-1}(\mathcal C)\to \mathcal C$. 
Note that in  \cite[Theorem 3.16]{BF}, the line bundle is the dual
of our $\mathcal{L}$ and the spin structure 
of $\mathcal C$ (identified with $S^1$) is trivial. In our situation
the spin structure of $\mathcal C$ is non-trivial, due to the 
choice of spinor 
$\Lambda (T^{*(0,1)}\sigma^{-1}(\mathcal C))\otimes E$,
with $E=L^k\otimes K^s$, in Section \ref{CSform}.
Thus the factor $(-1)^{{\rm Ind} D_{+}}$ 
from \cite[(3.164)]{BF} does not appear in 
(\ref{eq:gCS61}).
Moreover, to apply \cite[Theorem 1.15]{BGS},
we assume the metric $g$ on $T\Sigma$
is the restriction of a K\"{a}hler metric 
on a neighborhood  of $\sigma^{-1}(\mathcal C)$.
 For a detailed discussion on  
    how to combine Bismut-Gillet-Soul\'{e}'s curvature formula
    and Bismut-Freed's holonomy theorem to study the holonomy
    of the determinant line bundle even for a singular fibration
    see \cite[\S6]{BB}.
    
The eta invariant (form) is closely related to Chern-Simons theory.
 Let $E$ be a complex vector bundle on a manifold $X$.
Let $\nabla ^E$ be a connection on $E$, 
with curvature $\mathrm F^E=(\nabla ^E)^2\in
\Omega^2(X, \End(E))$.
For any real polynomial $q\in\R[z]$, we set
\begin{align}\label{curcm6.4}
 Q\big(\mathrm F^E\big)= \tr\Big[q\big(\tfrac{\sqrt{-1}}{2\pi} 
 \mathrm F^E\big)\Big]
\in\Omega^{\bullet}(X,\C)\,.
\end{align}    
By the Chern-Weil theory \cite[Theorem B.5.1]{MM},
$Q\big(\mathrm F^E\big)$ is a closed 
differential form and its de Rham cohomology class 
$[Q\big(\mathrm F^E\big)]\in H^{2\bullet}(X,\C)$
does not depend on the choice of $\nabla ^{E}$.

The Chern-Simons class 
$\widetilde{Q}(\nabla ^E_0, \nabla ^E_1)
\in \Omega^{2\bullet-1} (X,\C)/d\Omega^{\bullet}(X,\C)$
 associated to  two connections
$\nabla ^E_0$ and $\nabla ^E_1$ on $E$
is well defined (cf. \cite[Theorem B.5.4]{MM}), and
\begin{equation}\label{curcm6.15}
d \widetilde{Q}(\nabla ^E_0, \nabla ^E_1)
= Q(\mathrm F^E_1)-Q(\mathrm F^E_0).
\end{equation}

Let's go back to the context of \eqref{eq:gCS61}.
Let $g_{1}, h_{1}^{L}$ be another couple of metrics on 
$T\Sigma, L$, and 
$\overline{\eta}_{1}$ be the adiabatic limit of the associated 
reduced eta invariant. Then by 
\cite[(3.196)]{BF},  by applying
\cite[(6.33)-(6.36)]{BB} for the fibration 
$\sigma^{-1}(\mathcal C)\times [0,1]\to \mathcal C\times [0,1]$,
we have modulo $\Z$
$$\overline{\eta}-\overline{\eta}_1
= \int_{\sigma^{-1}(\mathcal C)}\text{\rm Chern-Simons classes 
associated with two pairs  $g, h$ 
and $g_{1}, h_{1}^{L}$.}$$
More precisely, let $g_{t}, h_{t}^{L}$ ($t\in [0,1]$)
be a path of metrics from $g, h$  to 
$g_{1}, h_{1}^{L}$. 
Let $A_{t}= \partial \log ((h^k)^{-1}h_{t}^{L^k})$,
and $B_{t}= \partial \log ((g)^{-1}g_{t})$,
and $F_{0}, R_{0}$ the curvatures of the Chern connections
on $(L, h^{k}), (T\Sigma, g)$,
then modulo $2\pi i\Z$,
\begin{equation}\label{curcm6.20}
\begin{split}
2\pi i(\overline{\eta}- \overline{\eta}_{1})
&=-\frac{i}{4\pi}\int\limits_{\sigma^{-1}(\mathcal{C})}\int_{0}^{1} dt 
\Big[2 \Big(\frac{\partial}{\partial t}A_{t}\Big) 
(\mathrm F_{0}+ d A_{t})\\
&\qquad\qquad\qquad\qquad+ (1-2s)\Big(\Big(
\frac{\partial}{\partial t}A_{t}\Big)\wedge (\mathrm R_{0}+ d B_{t})
+(\mathrm F_{0}+ d A_{t})\wedge
\Big(\frac{\partial}{\partial t}B_{t}\Big) \Big)\\ 
&\qquad\qquad\qquad\qquad
+  2 \left(\frac{(1-2s)^2}4-\frac1{12}\right)
\Big(\frac{\partial}{\partial t}B_{t}\Big)
\wedge (\mathrm R_{0}+ d B_{t}) \Big]\\
&=\frac{i}{4\pi}\int\limits_{\sigma^{-1}(\mathcal{C})} A\wedge dA
+\frac{1-2s}2(A\wedge d\omega+dA\wedge\omega)+
\left(\frac{(1-2s)^2}4-\frac1{12}\right)\omega\wedge d\omega\,\\
&\qquad- (A\to A_{1}, \omega\to \omega_{1}).
\end{split}
\end{equation}
This ends the proof of Theorem \ref{prop2}. \hfill $\square$

\smallskip

\begin{rem}
Let us briefly recall the classical Chern-Simons functional.
Let $Z$ be a manifold of dimension $\leq 3$. Then any
principal $G$-bundle of a simply connected compact Lie group $G$ 
over $Z$ is trivializable.  Assume also the structure group of 
$E$ is a simply connected compact Lie group, then $E$
is a trivial vector bundle on $Z$, and once we fix
a trivialization of $E$, it induces a connection 
$\nabla ^E_0$ which is the usual differential $d$. Then we 
denote $\widetilde{Q}(\nabla ^E_0, \nabla ^E_1)$ 
simply by $\widetilde{Q}(\nabla ^E_1)$.
In particular, if $Z$ is an oriented $3$-dimensional manifold,
this implies that the tangent bundle $TZ$ is trivial
and we fix a trivialization of $TZ$. 
The classical Chern-Simons functional 
 associated with the connection  $\nabla^{TZ}=d+A$ is 
 \cite[Example B.5.7]{MM}
 \begin{align}\label{curcm7.36}
CS(\nabla^{TZ})=\int_{Z}\widetilde{p}_{1}(d, \nabla^{TZ})
= \frac{-1}{8\pi^2} \int_{Z}
\tr\!\left[A\wedge dA 
+ \frac{2}{3}A\wedge A\wedge A\right]\in \R,
\end{align}
where $p_{1}$ is the first Pontryagin form.
Once we use another trivialization of $TZ$,
$CS(\nabla^{TZ})$ can change by an integer, thus 
$CS(\nabla^{TZ})\in \R/\Z$ does not depend on the trivialization
of $TZ$ and it is a well-defined functional on the space 
of connections on $TZ$.

However, in the context of our paper,
for $Z=\sigma^{-1}(\mathcal{C})$, the associated Lie group
is $S^1$ which is not simply connected. The line bundles
$L$ and $T\Sigma$ are not trivial on $Z$,
thus the last two terms in \eqref{curcm6.20} are not well-defined
Chern-Simons classes, but their difference is well defined
on  $\sigma^{-1}(\mathcal{C})$.

We show here how to define rigorously
the last two terms in \eqref{curcm6.20}.
By \cite[\S VII, Theorem 2]{Kirby89}, any connected orientable 
compact $3$-dimensional manifold $Z$ bounds a simply connected 
orientable $4$-dimensional manifold $W$. This implies that any 
smooth complex line bundle on $Z$ can be extended 
to a smooth complex line bundle on $W$.

Now let $W$ be such a $4$-dimensional manifold
with boundary the $3$-dimensional manifold 
$Z=\sigma^{-1}(\mathcal{C})$. We extend the line bundles
$L$ and $K_{\Sigma}^{1/2}$ from $\sigma^{-1}(\mathcal{C})$
to line bundles on $W$ denoted by $L$ and $K^{1/2}$, respectively. 
We extend also the connections of $L,K_{\Sigma}^{1/2}\to Z$
to connections of $L,K^{1/2}\to W$ such that
they have a product structure near the boundary $Z$, that is, 
on a neighborhood $Z\times [0, 1]$ of $Z$ in $W$,
they are pull-back of the corresponding connections on $Z$.
We denote the curvatures of $L$ and $K^{1/2}$ 
by $F$ and $-\frac{1}{2}R$, respectively. 
Inspired by \cite[(2.3)]{W05}, we can simply define 
the Chern-Simons functional by
\begin{align}\label{curcm6.37}
CS(A, \omega)= 
-\frac{1}{8\pi^2}\int_{W}\left[\mathrm{F}\wedge \mathrm{F}
+(1-2s)\, \mathrm{F}\wedge \mathrm{R}
+\left(\frac{(1-2s)^2}4-\frac1{12}\right)\,
\mathrm{R}\wedge \mathrm{R}\right].
\end{align}
Then $CS(A, \omega)$ is well-defined modulo $\frac{1}{6}\Z$, 
i.e., it does not depend on the choice of $W$ and the 
extension of the line bundles and connections modulo 
$\frac{1}{6}\Z$.  Let $(W', L, K^{1/2})$ be another 
triple of extensions (with connections). Then 
 $W\cup (-W')$ is a closed manifold (where $-W'$ has
  the same underlying space
 as $W'$ but the opposite orientation) and 
 $L, K^{1/2}$ glue together to complex  line bundles on
 $W\cup (-W')$
 with induced smooth connections and curvatures 
 $F$ and $-\frac{1}{2}R$, thus  
 \begin{align}\label{curcm6.40}\begin{split}
&-\frac{1}{4\pi^2}\int_{W\cup (-W')} 
 \mathrm{F}\wedge \mathrm{F} = 
 \int_{W\cup (-W')}  c_{1}(L)^2\in \Z,\\
& -\frac{1}{8\pi^2}\int_{W\cup (-W')} 
 \mathrm{F}\wedge \mathrm{R} = 
 \int_{W\cup (-W')}  c_{1}(L) c_{1}(K^{1/2})\in \Z,\\
& -\frac{1}{16\pi^2}\int_{W\cup (-W')} 
 \mathrm{R}\wedge \mathrm{R} = 
 \int_{W\cup (-W')}  c_{1}(K^{1/2})^2\in \Z.
\end{split}\end{align}
Thus for $s\in \frac{1}{2}\Z$, $CS(A, \omega)$
is well-defined modulo $\frac{1}{6}\Z$.

Note that by removing the term $-\frac{1}{12}R\wedge R$
from the right-hand side of \eqref{curcm6.37} the resulting 
expression is well defined modulo $\frac{1}{2}\Z$.

Note also that if $W$ and $W'$ are fibrations of Riemann 
surfaces over surfaces $G,G'$ with boundary $\mathcal C$, 
we can choose $(K^{1/2})^2$ as the cotangent bundle
of the Riemann surfaces. Then we can apply the 
Atiyah-Singer family index theorem for the fibration 
$\sigma: W\cup (-W') \to G\cup G'$,
and we have 
 \begin{align}\label{curcm6.37-1}
\frac{1}{8\pi^2}\int_{W\cup (-W')}
\left[\mathrm{F}\wedge \mathrm{F}
+(1-2s)\, \mathrm{F}\wedge \mathrm{R}
+\left(\frac{(1-2s)^2}4-\frac1{12}\right)\,
\mathrm{R}\wedge \mathrm{R}\right]\in \Z.
\end{align}
We can reformulate \eqref{curcm6.20} rigorously as
\begin{equation}\label{curcm6.38}
   \overline{\eta}- \overline{\eta}_{1}
=CS(A, \omega) -CS(A_{1}, \omega_{1})\:\:
\text{modulo $\frac{1}{6}\Z$}.
\end{equation}
Finally, we shall also add that it would be interesting to understand whether we have individually 
in \eqref{curcm6.38},
\begin{equation}\label{curcm6.381}
\overline{\eta} =CS(A, \omega)\:\:\text{modulo $\frac{1}{6}\Z$}.
\end{equation}
\end{rem}

In summary, we have derived the adiabatic phase starting
from the 2d QH-state and its generating functional, and using
the Bismut-Gillet-Soul\'e curvature formula for Quillen metric.
This observation establishes the link between 2d
\cite{K,CLW,FK,CLW1,KW}
and (2+1)d \cite{AG1,AG2,AG4,WZ,FS} approaches to QHE.

\section*{Appendix A: Quillen metric and analytic torsion}
Here we will review the basic facts on the Quillen metric
cf. \cite[\S 5.5]{MM}. Let $X$ be a compact complex manifold
and let $E$ be a holomorphic vector bundle on $X$.
Let $\Omega^{(0,\bullet)}(X, E)= \bigoplus_j \Omega^{(0,j)}(X, E)$
be the space of smooth anti-holomorphic forms on $X$
with values in $E$.
The holomorphic structure on $E$ allows us to define the
Cauchy-Riemann operators $\overline{\partial}^E$
as follows.
Any section $s\in \cC^\infty(X,E)$ has the local form
$s=\sum_{l} \varphi_{l}\xi_{l}$ where $\{\xi_{l}\}_{l=1}^m$
is a local holomorphic frame of $E$ and $\varphi_{l}$
are smooth functions. We set
$\bp^E s=\sum_{l}(\bp \varphi_{l})\,\xi_{l}$,
where $\bp \varphi_{l}= \sum_j d{\overline z}_j
\frac{\partial }{\partial{\overline z}_j}\varphi_{l}$
in holomorphic coordinates $(z_1,\cdots,z_n)$.
We define the $\overline{\partial}^E$-operator
on $\Omega^{(0,\bullet)}(X, E)$ by setting
 $\bp^E (\alpha\otimes s)= \bp\alpha \otimes s
 + (-1)^{\deg \alpha}  \alpha \wedge \bp^E s$,
 where $\alpha\in\Omega^{(0,\bullet)}(U)$
 and $s\in \cC^\infty(U,E)$ for some open set $U\subset X$.
We have then $( \bp^E)^2=0$. Thus we get thus
the Dolbeault complex
 $(\Omega^{(0,\bullet)}(X, E), \bp^E)$ whose cohomology
is called Dolbeault cohomology,
 $$H^{(0,\bullet)}(X, E) = \bigoplus_{j=0}^n H^{(0,j)}(X, E)
\quad {\rm with  }\, \,  H^{(0,j)}(X, E)
= \ker \bp^E|_{\Omega^{(0,j)}}/
 {\rm Im}  \bp^E|_{\Omega^{(0,j-1)}}. $$
We denote it simply
by $H^{\bullet}(X, E) := \bigoplus_j H^{j}(X, E)
:=\bigoplus_{j=0}^n H^{(0,j)}(X, E)$.

Let $h^{TX}$ be a Hermitian metric on the holomorphic
tangent bundle $TX$ and $h^E$ a Hermitian metric on $E$.
They induce an $L^2$-inner product $\langle\cdot\,,\cdot\rangle$
on $\Omega^{(0,\bullet)}(X, E)$ as in \eqref{inner}.
Let $\bp^{E*}$ be the adjoint of $ \bp^E$ with respect to
this $L^2$-inner product.
The Kodaira Laplacian $\square^E$ is defined by
$$\square^E=(\bp^{E}+ \bp^{E*})^2
= \bp^{E*}\bp^{E}+ \bp^{E}\bp^{E*}:
\Omega^{(0,j)}(X, E)\to \Omega^{(0,j)}(X, E).$$
By Hodge theory, the map
$\ker \square^E|_{\Omega^{(0,j)}}\to H^{j}(X, E)$,
$\sigma\mapsto[\sigma]$,
which sends a harmonic form to its Dolbeault cohomology class is
an isomorphism
$\ker \square^E|_{\Omega^{(0,j)}}\simeq H^{j}(X, E)$.
We denote  by $\square^{E, >0}$ the restriction of
$\square^E$ to the orthogonal complement of $\ker \square^E$.
Then for $u \in \C$, ${\rm Re}(u)> n$,
 the operator $(\square^{E, >0})^{-u}$
is a trace class operator and the theta function
\begin{equation}\label{eq2.}
\theta(u)= -
\sum_{j=0}^n (-1)^j j\,  \tr|_{\Omega^{(0,j)}}
\Big[ (\square^{E, >0})^{-u} \Big]
\end{equation}
extends as a meromorphic function on $\C$ which is
holomorphic at $0$. This is a simple application
of the small time asymptotic expansion of the heat
kernel and the Mellin  transformation.
The Ray-Singer analytic torsion is defined as
\begin{equation}\label{antor}
T =  \exp\left(-\frac{1}{2} \frac{\partial\theta}{\partial u}(0)\right).
\end{equation}
The determinant line of the cohomology $H^{\bullet}(X,E)$
is the complex line given by
$$\det H^{\bullet}(X,E)
= \bigotimes_{j=0}^n (\det H^j (X, E))^{(-1)^j},$$
where for $F$ a complex vector space,
$\det F= \Lambda^{\max} F$ is the complex line given by the
exterior algebra of maximum degree of $F$ and
$(\det F)^{-1}$ the dual line of $\det F$. We define
\begin{align}\label{apb5.1}
\lambda(E)= (\det H^{\bullet}(X,E))^{-1}.
\end{align}
Let $h^{H(X,E)}$ be the $L^2$-metric on $H^{\bullet}(X,E)$
induced by the $L^2$-scalar product on $\Omega^{0,\bullet}(X,E)$.
via the Hodge theory. 
Let $|\cdot|_{\lambda(E)}$ be the $L^2$-metric on $\lambda(E)$
induced by $h^{H(X,E)}$.

\begin{maindefn} \label{apbt5.3} The Quillen metric
$\|\cdot\|_{\lambda(E)}$ on the complex line $\lambda(E)$
is defined by
\begin{align}\label{apb5.2}
\|\cdot\|_{\lambda(E)}= |\cdot|_{\lambda(E)}
\exp\Big(\!-\frac{1}{2}\theta^\prime(0)\Big).
\end{align}
\end{maindefn}

The Quillen metric $\|\cdot\|_{\lambda(E)}$ depends on
the choice of the metrics $h^{TX}$ and $h^E$.
If $h^{TX}_0, h_1^{TX}$ are two K\"ahler metrics on $TX$,
and $h^E_0, h_1^E$ are two Hermitian metrics on $E$,
let $\|\cdot\|_{\lambda(E), i}$ be the Quillen metric on $\lambda(E)$
associated with $(h^{TX}_i, h^E_i)$ for $i=0,1$.
Bismut-Gillet-Soul\'{e} \cite[Theorem 1.23]{BGS}
proved the anomaly formula, by expressing
$$\log\frac{\|\cdot\|_{\lambda(E), 0}^2}
{\|\cdot\|_{\lambda(E), 1}^2}$$
as an integral of certain explicit computable local terms,
the so-called Bott-Chern secondary classes. In particular,
Theorem \ref{prop11} in Section \ref{sec4} is a special case of
their result.

Let $\pi: M\to B$ be a holomorphic submersion of
complex manifolds with compact fiber $X$,
and let $E$ be a holomorphic vector bundle on $M$.
Usually,  the dimension of the fiberwise Dolbeault cohomology
$H^j(X_b, E|_{X_b})$ can jump dramatically for $b\in B$,
thus  it is not clear whether the family of complex lines
$\lambda(E)_b$ parameterized by $b\in B$ can form a
smooth line bundle on $B$.

Following a suggestion of Grothendieck,
Knudsen-Mumford \cite{KnuM76}
solved the problem in 1976 for a projective map $\pi: M\to B$
(which means that there exists a holomorphic
vector bundle $F$ over $B$ and a closed immersion
 $\imath: M\to \mathbb{P}(F)$ such that $\pi= p\circ \imath$,
 where $p: \mathbb{P}(F)\to B$ is the projectivisation of
 $F\to B$).
They defined namely in this case a holomorphic line bundle
$\lambda^{KM}$ over $B$ such that on $b\in B$,
we have canonically $\lambda^{KM}_b\simeq \lambda(E)_b$.

By extending \cite{Q}, Bismut-Freed \cite{BF} shown that
we can define analytically  a smooth structure of $\lambda(E)$
over $B$.

Assume now $\pi: M\to B$ is locally K\"ahler, i.e.,
if there exist an open covering $\mathcal{U}$ of $B$ such that
for any $U\in \mathcal{U}$,
there exists a K\"{a}hler metric on $\pi^{-1}(U)$.
Bismut-Gillet-Soul\'{e} \cite[\S 1 i)]{BGS} defined
a holomorphic structure on $\lambda(E)$ over $B$
such that  $\lambda(E)$ is a holomorphic line bundle on $B$.
Let $TX$ be the relative holomorphic tangent bundle
of $\pi: M\to B$ and let $h^{TX}$ be a smooth Hermitian metric
on $TX$ which induces a K\"ahler metric on each fiber $X$.
 Let $h^E$ be a Hermitian metric on
 $E$. Since the fiberwise cohomology group
 $H^j(X_b, E|_{X_b})$  might jump, both the $L^2$-metric
 $|\cdot|_{\lambda(E)_b}$ on the line bundle $\lambda(E)$
 and the analytic torsion $T_b$ need not be smooth over $B$.
 However, Bismut-Gillet-Soul\'{e}
 proved that the Quillen metric $\|\cdot\|_{\lambda(E)_b}$
forms a smooth metric on the line bundle
 $\lambda(E)$ over $B$, and in \cite[Theorems 1.9, 1.27]{BGS},
 they computed the curvature of the Chern connection
 of the holomorphic Hermitian line bundle
 $(\lambda(E), \|\cdot\|_{\lambda(E)})$ over $B$
by using the local family
Atiyah-Singer index theorem of Bismut \cite{B}.
If $\pi: M\to B$ is projective, they also proved that
there is a canonical isomorphism
$\lambda^{KM}\simeq\lambda(E)$ as holomorphic line bundles
on $B$.

Let $\omega^M$ be a closed real $(1,1)$-form on $M$ such that
its restriction along each fiber $X$ defines
a hermitian metric $h^{TX}$
(In particular,  $\pi:M\to B$ is locally K\"{a}hler).
We assume that for each $j$, the dimension of
$H^j(X_{b},E_{X_{b}})$ is constant on $b\in B$. Then
$H^j(X_{b},E_{X_{b}})$ coincides with the fiber of the $j$-th
direct image $R^j\pi_{*}E$ of the sheaf of holomorphic sections of
$E$ via $\pi$, thus from the holomorphic structure on
 $R\pi_{*}E$, they form a holomorphic vector bundle
 $H^j(X,E)$ on $B$.

 Let $\nabla^{H(X,E)}$ be the Chern connection on the holomorphic
 Hermitian vector bundle on $(H^\bullet(X,E)$,
 $h^{H(X,E)})$ over $B$, cf. 
 \cite[Theorem 1.1.5]{MM}.
 A natural question is how to compute
  $\nabla^{H(X,E)}$ by using the
  Chern connection $\nabla^{E}$ on $(E,h^{E})$.
  Observe that for any $s\in \mathcal{C}^\infty(B, H^j(X,E))$,
  we can identify it as an element of
 $ \mathcal{C}^\infty(M, \Lambda^{j}(T^{*(0,1)}X)\otimes E)$
 via the Hodge theory.

 Let $P_{b}: \Omega^{(0,\bullet)}(X_{b}, E)\to \ker \square^{E}$
 be the fiberwise orthogonal projection for $b\in B$.
 Let $\nabla^{\Lambda(T^{*(0,1)}X)\otimes E}$
 be the connection on $\Lambda(T^{*(0,1)}X)\otimes E$ induced
 by the Chern connections on $(TX, h^{TX})$,
 $(E, h^{E})$. Let
 \begin{align}\label{eq5:}
T^H_{\R}M:= \{ V\in T_{\R}M:
\omega^M(V,Y)=0 \, \, \text{ for  all } X\in T_{\R}X\}.
\end{align}
Here we add a subscript $\R$ to indicate the underlying real
vector bundles. Then $T_{\R}M= T^H_{\R}M\oplus T_{\R}X$
and for $U\in T_{\R, b}B$, we denote by $U^H\in T^H_{\R}M$
the unique lift such that $d\pi (U^H)=U$. If $u\in T^{(1,0)}B$ 
then $u^H\in T^{(1,0)}M$. Then
Bismut-K\"{o}hler \cite[Theorem 3.5]{BKo} proved that
for any $s\in \mathcal{C}^\infty(B, H^\bullet(X,E))$,
\begin{align}\label{eq6:}
\nabla^{H(X,E)}_{U}s =
P\nabla^{\Lambda(T^{*(0,1)}X)\otimes E}_{U^H} s.
\end{align}
In particular for $s\in \mathcal{C}^\infty(B, H^0(X,E))$, we have
\begin{align}\label{eq8:}
\nabla^{H^0(X,E)}_{U}s = P\nabla^{ E}_{U^H} s.
\end{align}
To verify directly that the right hand side
of (\ref{eq6:}) is a Hermitian connection,
we need to verify that for
$s_{1}, s_{2}\in \mathcal{C}^\infty(B, H^\bullet(X,E))$
\begin{align}\label{eq7:}
U\left\langle  s_{1}, s_{2}\right\rangle=
\left\langle  P\nabla^{\Lambda(T^{*(0,1)}X)\otimes E}_{U^H}s_{1},
s_{2}\right\rangle +         \left\langle  s_{1},
P\nabla^{\Lambda(T^{*(0,1)}X)\otimes E}_{U^H}s_{2}\right\rangle,
\end{align}
but this is a consequence of
the horizontal Lie derivative of the fiberwise volume form
$(\omega^M|_{X}) ^n/n!$ is zero, i.e.,
\begin{equation}\label{lie}
L_{U^H}(\omega^M|_{X}) ^n=0
\end{equation}
which is implied by the closeness of $\omega^M$
cf. \cite[Theorem 1.7, part d]{BGSII}.

\section*{Appendix B: Expansion of the Bergman kernel on
Riemann surfaces}
Here we review the asymptotic expansion of the Bergman kernel,
following Ref.\ \cite{MM}.

Let $(X, J)$ be a compact complex manifold of
$\dim_{\C}X=n$. Let $\omega$ be a positive $(1,1)$-form on $X$,
 and we denote by $g^{TX}$
the Riemannian metric associated to $\omega$.
Then $dv_X= \om^n/n!$ is the Riemannian volume
form of $(X, \om)$.
Let $\nabla^{TX}$ be the Levi-Civita connection on $(X,\omega)$.
We denote by
$R^{TX}=(\nabla^{TX})^2$ the curvature, by $\ric$
the Ricci curvature and
by $R$ the scalar curvature of $\nabla^{TX}$.
Let $\ric_\om= \ric(J\cdot,\cdot)$ be the $(1,1)$-form
associated to $\ric$.
Let $\Delta$ be the (negative) Laplace operator on $(X, g^{TX})$
acting on the functions on $X$.

Let $(L, h^L)$ and $(E, h^E)$ be  holomorphic Hermitian vector
bundles on $X$, with $\operatorname{rk}L=1$.
On the space $\cC^\infty (X,L^k\otimes E)$ of smooth sections of
$L^k\otimes E$ we introduce an
$L^2$--scalar product $\langle\,\cdot\,,\cdot\,\rangle$ given by
\begin{align}\label{n2}
&\langle \sigma_1,\sigma_2 \rangle
:=\int_X\langle \sigma_1(x),\sigma_2(x)
\rangle_{h^{L^k}\otimes h^E}\,dv_{X}(x)\,.
\end{align}
We denote by $L^2(X,L^k\otimes E)$ the completion of
$\cC^\infty (X,L^k\otimes E)$
with respect to \eqref{n2}. We also introduce the space
$H^0(X,L^k\otimes E)$
of holomorphic sections of $L^k\otimes E$ on $X$.
Let
\begin{equation}\label{n2,1}
P_k:L^2(X,L^k\otimes E)\longrightarrow H^0(X,L^k\otimes E)
\end{equation}
be the orthogonal projection, called Bergman projection.
For $z,z'\in X$ let $P_k(z,z')$  be the smooth kernel of
$P_k$ with respect to $dv_{X} (x')$, that is,
\begin{equation}\label{n3}
(P_ks)(z)=\int_X P_k(z,z')s(z')dv_{X} (z')\,,\quad
s\in L^2(X,L^k\otimes E).
\end{equation}
Note that the $L^2$--scalar product
$\langle\,\cdot\,,\cdot\,\rangle$
and hence the kernel $P_k(\,\cdot\,,\cdot\,)$ depend on the
Hermitian metrics $h^L$, $h^E$ and $\omega$.
Let $\{s_{j}\}_{j=1}^{d_{p}}$ be an orthonormal basis of
$\big(H^0(X,L^k\otimes E),\langle\,\cdot\,,\cdot\,\rangle\big)$.
If $E$ is a line bundle,
 \begin{equation}\label{1d0}
P_k(x,x)= \sum_{k=1}^{d_{p}}\big|s_{k}(x)\big|^2_{h^{L^k}
\otimes h^E}\,.
\end{equation}

For $\psi\in\cC^\infty(X)$ and denote
\begin{equation} \label{n4}
h^L_\psi:=h^L e^{-\psi}.
\end{equation}
We assume that the curvature form of $(L,h^L_\psi)$, 
cf. \cite[Theorem 1.1.5]{MM}
\begin{equation} \label{n5}
\omega_\psi:=\frac{i}{2 \pi} R^L_\psi
:=\frac{i}{2 \pi} R^L
+\frac{i}{2 \pi}\partial\overline\partial\psi
\end{equation}
is a K\"ahler metric on $X$.

\noindent
We will associate to a $(1,1)$-form $\eta$ an endomorphism
\begin{equation} \label{n6}
\dot{\eta}\in\End(T^{(1,0)}X)\,,\:\:\eta(u,\overline{v})=
\langle\dot{\eta}(u),v\rangle_{\om}.
\end{equation}
Note that the volume forms corresponding to $\omega$
and $\omega_\psi$ are linked by
\begin{equation} \label{n7a}
\omega_\psi^n/n!=(\det\dot{\omega}_\psi)\omega^n/n!\,.
\end{equation}
We consider the $L^2$-scalar product
$\langle\,\cdot\,,\cdot\,\rangle$ on $\cC^\infty(X,L^k\otimes E)$
constructed as in \eqref{n2} but with the metrics $h^L_\psi$ on
$L$, $h^E$ on $E$ and with respect to the volume form
$\omega^n/n!$ on $X$. We have then a corresponding
Bergman projection $P_k$ as in \eqref{n2,1}.
We denote by $P_k(x,x')$ the Schwartz kernel of $P_k$
with respect to the volume form $\omega^n/n!$ (see \eqref{n3}).

We consider moreover the $L^2$--scalar product
$\langle\,\cdot\,,\cdot\,\rangle_{\psi}$
constructed as in \eqref{n2} with the metrics $h^L_\psi$ on
$L$, $h^E_\psi=(\det\dot{\omega}_\psi)^{-1}h^E$ on $E$ and
with respect to the volume form
$\omega_\psi^n/n!$ on $X$.
Note that $\langle\,\cdot\,,\cdot\,\rangle_\psi
=\langle\,\cdot\,,\cdot\,\rangle$ so that
the Bergman projection $P_k$ (cf.\ \eqref{n2,1}) is the same for
these scalar products. We denote $P_{k,\psi}(x,x')$
the Schwartz kernel by of
$P_k$ with respect to $\omega_\psi^n/n!$.
By \eqref{n7a}, we have
(cf. \cite[(4.1.114)]{MM} or \cite[Remark 0.5]{MM1})
\begin{equation} \label{n7}
P_k(x,x')=(\det\dot{\omega}_\psi)(x')P_{k,\psi}(x,x').
\end{equation}
The advantage of passing from $P_k$ to $P_{k,\psi}$
is that $(L,h^L_\psi)$ polarizes $(X,\omega_\psi)$,
that is, $\frac{i}{2\pi}R^{(L,h^L_\psi)}=\omega_\psi$,
cf.\ \eqref{n5}.
We will use now the asymptotic expansion of
the Bergman kernel function $P_{k,\psi}(x,x)$.
For this purpose we introduce more notations.

In the sequel we will mainly work with the K\"ahler form
$\om_\psi$, in which case we will use
a subscript 
 $\psi$, e.\,g., $\ric_\psi$, $\Delta_\psi$ etc.
Then for any $f\in \cC^{\infty}(X)$ we have
$\left\langle  \ov{\partial}\partial f, \om_{\psi}\right\rangle_{\psi}
= \frac{i}{2}\Delta_{\psi} f$.
Set
\begin{align} \label{bk2.5}
R^E_\psi=R^{(E,h^E_\psi)},\quad R^E_{\psi, \Lambda}
=\left \langle R^E_\psi, \om_\psi\right \rangle_{\psi}\, .
\end{align}
By \cite[Theorem 0.1]{MM1}
 (cf. \cite{C}, \cite{Z}, \cite{Lu}  without the twist bundle $E$),
we have the following asymptotics as $k\to\infty$,
\begin{equation} \label{n8}
P_{k,\psi}(x,x)=k^n+\left(\frac{R_\psi}{8\pi}+
\frac{i}{2\pi}R^E_{\psi,\Lambda}
\right)k^{n-1}
+\frac{1}{\pi^2}\big(\bb_{2\C\psi}+\bb_{2E\psi}\big)k^{n-2}
+\mathcal{O}(k^{n-3}),
\end{equation}
where
\begin{equation}\label{abk2.6}\begin{split}
\bb_{2\C\psi}= &  \frac{\Delta_\psi R_\psi}{48}
+ \frac{1}{96}|R^{TX}_\psi|^2 _{\psi}
- \frac{1}{24} |\ric_\psi|^2_{\psi} +  \frac{1}{128} R^2_\psi\\
=&  \frac{\Delta_\psi R_\psi}{48}
+ \frac{1}{6} R_{\psi,j\ov{\ell} m\ov{q}}  R_{\psi,\ell\ov{j}q\ov{m}}
- \frac{2}{3} R_{\psi,\ell\ov{\ell} m\ov{q}}  R_{\psi,j\ov{j}q\ov{m}}
+  \frac{1}{2} R_{\psi,\ell\ov{\ell}q\ov{q}} R_{\psi,j\ov{j}m\ov{m}},
\end{split}\end{equation}
and
\begin{multline}\label{abk2.7}
\bb_{2E\psi}
=\frac{i}{16} \Big( R_\psi R^E_{\psi,\Lambda}
    - 2 \langle\ric_\psi, R^ E_\psi\rangle_{\psi}
    +\Delta^E_\psi  R^E_{\psi,\Lambda}\Big)
    - \frac{1}{8} (R^E_{\psi,\Lambda})^2 +
     \frac{1}{8}\langle R^ E_\psi, R^ E_\psi\rangle_{\psi}\\
     = R^E_{\psi,q\ov{q}}  R_{\psi,j\ov{j}m\ov{m}}
- R^E_{\psi,m\ov{q}}  R_{\psi,j\ov{j}q\ov{m}}
+ \frac{1}{2}\left(R^E_{\psi,q\ov{q}} R^E_{\psi,m\ov{m}}
- R^E_{\psi,m\ov{q}} R^E_{\psi,q\ov{m}}\right)
+\frac{1}{2}R^E_{\psi,j\ov{j}\,;\,m\ov{m}}.
\end{multline}
Here in normal coordinate associated with $\omega_{\psi}$
at $x_{0}$,
\begin{equation} \label{lm01.4} \begin{split}
   & \om_{\psi}= \frac{i}{2}
    \sum_{j}dz_{j}\wedge d \ov{z}_{j},\\
&R_{\psi,j\ov{m}\ell\ov{q}} =\left \langle  R^{TX}_{\psi}
\Big(\frac{\partial}{\partial z_j},
\frac{\partial}{\partial \ov{z}_m}\Big) \frac{\partial}{\partial z_\ell},
\frac{\partial}{\partial  \ov{z}_q}\right\rangle_{\psi,x_0},\quad
R^E_{\psi,j\ov{\ell}}= R^E_{\psi,x_0}
\Big(\frac{\partial}{\partial z_j},
\frac{\partial}{\partial \ov{z}_\ell}\Big),\\
&R^E_{\psi, j\ov{q}; t\ov{s}}=
\frac{\partial^2}{\partial z_{t}\partial \ov{z}_{s}}R^E_{\psi}
\Big(\frac{\partial}{\partial z_j},
\frac{\partial}{\partial \ov{z}_q}\Big), \quad
i  R^E_{\psi,\Lambda} = 2 R^E_{\psi,j\ov{j}}.
\end{split}\end{equation}

From now on, we assume that $X=\Sigma$ is a (connected)
Riemann surface with the K\"{a}hler form $\om$.
We use the local normal
coordinate associated with $\om$ near $x_{0}$, then at $x_{0}$,
\begin{align}\label{eq:mm3.1}
\om= \frac{i}{2}dz\wedge d \overline{z},
\quad \Delta
= 4 \frac{\partial^2}{\partial z\partial \overline{z}}\,,
\quad (\Delta \varphi)\om
= - 2 i\,\overline{\partial} \partial \varphi.
\end{align}
At $x_{0}$, the scalar curvature $R$
of $(X, \om)$ is given by
\begin{align}\label{eq:mm3.3}
R= 4 R^{T^{(1,0)}X}\Big(
\frac{\partial}{\partial z}, \frac{\partial}{\partial \overline{z}}\Big).
\end{align}
Thus
\begin{align}\label{eq:mm3.2}
-\frac{i}{2} R \, \om= R^{T^{(1,0)}X}
= -R^{K_{X}}
=\overline{\partial} \partial \log |\sigma|^2,
\end{align}
where $\sigma$ is a local holomorphic frame of $T^{(1,0)}X$.

For the K\"ahler form  $\om_{\psi}$ in (\ref{n5}),
by  (\ref{eq:mm3.1}),  we have
\begin{align}\label{eq:mm3.8}
\Delta_{\psi}= a^{-1}\Delta,
\quad \text{with } a= \frac{\om_{\psi}}{\om}.
\end{align}

Thus by (\ref{eq:mm3.1}), (\ref{eq:mm3.2}) and \eqref{eq:mm3.8}
we get
$-\frac{i}{2} R_{\psi}\om_{\psi}
= -\frac{i}{2}  R\om+ \overline{\partial} \partial \log a$,
and
\begin{align}\label{eq:mm3.9}
R_{\psi} =  a^{-1}  R - \frac{1}{a}  \Delta \log a.
\end{align}
Observe that $\deg K_X$ is even, thus $K_X^{1/2}$
is well-defined. Now, for $s\in \frac{1}{2}\Z$,
we take $E= K_{X}^s$ with metric
$g^{\otimes s}$ induced by $\om$,
then $h^{E}_{\psi}= a^{-1} g^{\otimes s}$ and
\begin{align}\label{eq:mm3.14}
R^{E}_{\psi,\Lambda}\om_{\psi}= R^{E}_{\psi}=R^{K_{X}^s}_{\psi}
= s \, R^{K_{X}} - \overline{\partial} \partial \log a .
\end{align}
In particular,  by (\ref{eq:mm3.1}), (\ref{eq:mm3.2})
and (\ref{eq:mm3.14}), we get
\begin{align}\label{eq:mm3.15}
i a R^{E}_{\psi,\Lambda}= -\frac{s}{2} R
+ \frac{1}{2}\Delta \log a
\end{align}
By \eqref{n8}, $P_{k,\psi}(x,x')$ has an asymptotic expansion with
coefficients
\begin{align}\label{eq:mm3.17}
\frac{R_\psi}{8\pi}+
\frac{i}{2\pi}R^E_{\psi,\Lambda}
= \frac{1}{2\pi}\left(\frac{1}{4}  R_{\psi} - \frac{s}{2a}R
+ \frac{1}{2a}\Delta \log a \right)
= \frac{1}{2\pi}\left(-\frac{1}{4}  R_{\psi}
+ \frac{-s+1}{2 a}R\right),
\end{align}
and
\begin{equation}\label{eq:mm3.12}
\begin{split}
\bb_{2\mathbb{C}\psi}&+ \bb_{2E\psi}=
\frac{\Delta_\psi R_\psi}{48}
+ \frac{i}{16}\Delta_\psi  R^E_{\psi, \Lambda}\\
&= \frac{1}{24 a} \Delta \left(\frac{1}{2 a} R
- \frac{1}{2 a}  \Delta \log a\right)
+ \frac{1}{16 a }  \Delta\left ( -\frac{s}{2 a}R
+ \frac{1}{2a}\Delta\log a\right)\\
&= \frac{1}{48 a}\Delta \left( \frac{-3s+2}{2 a}R
+ \frac{1}{2a}\Delta \log a\right).
\end{split}
\end{equation}
Combining \eqref{n7}, \eqref{n8}, \eqref{eq:mm3.8},
 \eqref{eq:mm3.17} and \eqref{eq:mm3.12} we obtain as
 $k\to\infty$,
\begin{equation}\label{eq:mm3.11}
\begin{split}
P_k(x,x)=&ak+\frac{a}{2\pi}\left(-\frac{1}{4}  R_{\psi}
+ \frac{-s+1}{2 a}R\right)\\
&+ \frac{1}{48\pi^2}\Delta\left( \frac{-3s+2}{2 a}R
+  \frac{1}{2a}\Delta \log a\right)k^{-1}+\mathcal{O}(k^{-2}).
\end{split}
\end{equation}
Setting now $\omega_\psi=a\omega$, $B=2\pi ka$ and
$m=s$, we arrive at the expression in Eq.\ \eqref{expansion}.

\end{document}